\documentclass[a4paper,UKenglish,numberwithinsect]{lipics}
\usepackage{microtype}
\usepackage{wrapfig}

\usepackage{amssymb}
\usepackage{verbatim}

\usepackage{stmaryrd}
\usepackage{latexsym}
\usepackage{amssymb}
\usepackage{amsmath}
\usepackage{array}
\usepackage[all]{xy}
\usepackage{url}

\newcommand{\nat}{\mathsf{nat}}

\newcommand{\setsorts}{\mathcal{B}}

\newcommand{\setvar}{\mathcal{V}}
\newcommand{\F}{\mathcal{F}}
\newcommand{\Rules}{\mathcal{R}}

\newcommand{\FV}{\mathit{FV}}

\newcommand{\atype}{\sigma}
\newcommand{\btype}{\tau}
\newcommand{\ctype}{\rho}
\newcommand{\abasetype}{\iota}
\newcommand{\bbasetype}{\kappa}
\newcommand{\aterm}{s}
\newcommand{\bterm}{t}
\newcommand{\cterm}{u}

\newcommand{\avar}{x}
\newcommand{\bvar}{y}
\newcommand{\cvar}{z}

\newcommand{\afun}{\ensuremath{{f}}}  
\newcommand{\bfun}{\ensuremath{{g}}}

\newcommand{\asub}{\gamma}
\newcommand{\bsub}{\delta}

\newcommand{\clausevar}{\texttt{(var)}}
\newcommand{\clauseapp}{\texttt{(app)}}
\newcommand{\clauseabs}{\texttt{(abs)}}
\newcommand{\clausefun}{\texttt{(fun)}}
\newcommand{\clauserule}{\texttt{(rule)}}

\newcommand{\clausebeta}{($\beta$)}

\newcommand{\abs}[2]{\lambda #1 . \, #2}
\newcommand{\app}[2]{#1 \cdot #2}
\newcommand{\subst}[2]{#1 #2}
\newcommand{\up}[1]{#1^\sharp}

\newcommand{\typepijl}{\!\Rightarrow\!}
\newcommand{\ftypepijl}{\Rightarrow}
\newcommand{\functionpijl}{\Rightarrow}
\newcommand{\decpijl}{\!\Rightarrow\!}
\newcommand{\arrz}{\ensuremath{\rightarrow}}
\newcommand{\arr}[1]{\ensuremath{\rightarrow_{#1}}}

\newcommand{\dppijl}{\arrz}
\renewcommand{\c}{\mathtt{c}}

\newcommand{\sorsuccsim}{{\,{}_{_(}\!\!\succsim_{_)}}}

\newcommand{\gterm}{\succ}
\newcommand{\geqterm}{\succsim}
\newcommand{\geqorgterm}{\sorsuccsim}

\newcommand{\geqorg}{{\,{}_{_(}\!\!\geq_{_)}}}

\newcommand{\domain}{\mathsf{dom}}

\newcommand{\basealgebra}{\mathcal{A}}
\newcommand{\basealgebraset}{\mathit{A}}
\newcommand{\gwm}{\sqsupset}
\newcommand{\geqwm}{\sqsupseteq}
\newcommand{\WM}{\mathcal{WM}}
\newcommand{\algint}[1]{\llbracket #1 \rrbracket}
\newcommand{\lamalgint}[1]{[#1]}
\newcommand{\algintc}[1]{\llbracket #1 \rrbracket_{\constvaluation,\varvaluation}}
\newcommand{\constvaluation}{\mathcal{J}}
\newcommand{\varvaluation}{\alpha}
\newcommand{\varvaluationb}{\delta}
\newcommand{\fatlambda}{\lambda\!\!\!\lambda}
\newcommand{\N}{\mathbb{N}}

\newcommand{\Pol}{\mathit{Pol}}

\newcommand{\wanda}{\textsf{WANDA}}
\newcommand{\TTTT}{%
 \textsf{T\kern-0.2em\raisebox{-0.3em}T\kern-0.2emT\kern-0.2em\raisebox{-0.3em}2}%
}
\newcommand{\aprove}{\textsf{AProVE}}
\newcommand{\minisat}{\textsf{MiniSAT}}

\newcommand{\typedepth}{\mathit{order}}

\renewcommand{\o}{\mathsf{o}}
\newcommand{\suc}{\mathsf{s}}
\newcommand{\nul}{\mathsf{0}}

\newcommand{\map}{\mathsf{map}}
\newcommand{\nil}{\mathsf{nil}}
\newcommand{\cons}{\mathsf{cons}}
\newcommand{\natlist}{\mathsf{natlist}}
\newcommand{\funlist}{\mathsf{flist}}

\newcommand{\I}{\mathsf{I}}

\newcommand{\append}{\mathsf{append}}
\newcommand{\reverse}{\mathsf{reverse}}
\newcommand{\hshuffle}{\mathsf{shuffle}}

\newcommand{\collapse}{\mathsf{collapse}}
\newcommand{\diff}{\mathsf{diff}}
\newcommand{\minimum}{\mathsf{min}}
\renewcommand{\gcd}{\mathsf{gcd}}
\newcommand{\build}{\mathsf{build}}

\newcommand{\paragraaf}[1]{\medskip \noindent\textit{\textbf{#1}}}
\newcommand{\sparagraaf}[1]{\noindent \textit{\textbf{#1}}}
\newcommand{\commentaar}[1]{\emph{Comment:} #1}
\newcommand{\note}[1]{\emph{Note: #1}}

\newcommand{\confreport}[2]{#1}    % conference version
\renewcommand{\confreport}[2]{#2} % uncomment to see the report version instead

\title{Polynomial Interpretations for Higher-Order Rewriting
        \footnote{This research is supported by the Netherlands
          Organisation for Scientific Research (NWO-EW) under grant
          612.000.629 (HOT).}
      }

\author[1]{Carsten Fuhs}
\author[2]{Cynthia Kop}
\affil[1]{University College London, Gower Street,
          London WC1E 6BT, UK}
\affil[2]{Vrije Universiteit, De Boelelaan 1081a,
          1081 HV Amsterdam, The Netherlands}

\authorrunning{C.~Fuhs and C.~Kop}
\Copyright[nc-nd]{Carsten Fuhs and Cynthia Kop}

\serieslogo{rta-logo}%please do not change
\begin{document}

\maketitle

\begin{abstract}
The termination method of weakly monotonic algebras, which has been
defined for higher-order rewriting in the HRS formalism,
offers a lot of power, but has seen little use in recent years.
We adapt and extend this method to the alternative formalism of
\emph{algebraic functional systems}, where the simply-typed
$\lambda$-calculus is combined with algebraic reduction.
Using this theory, we define \emph{higher-order polynomial
interpretations}, and show how the implementation challenges of this
technique can be tackled.
A full implementation is provided in the termination tool \wanda.
\end{abstract}

\section{Introduction} \label{sec:introduction}

One of the most prominent techniques in termination
proofs for first-order term rewriting systems (TRSs)
is the use of \emph{polynomial
  interpretations}.  In this
method, which dates back to the seventies~\cite{lan:79}, terms are
mapped to polynomials
over (e.g.) $\N$.
The method is quite intuitive, since a TRS is usually
written with a meaning
for the function symbols
in mind, which can
often
be modeled by the
interpretation.
In addition, it has been implemented in
various automatic tools, such as \textsf{AProVE}
\cite{gie:sch:thi:06}, \TTTT\ \cite{kor:ste:zan:mid:09} and
\textsf{Jambox} \cite{end:12:1}.
Polynomial interpretations are an instance of the \emph{monotonic
algebra} approach \cite{end:wal:zan:08} which also includes for instance
\emph{matrix interpretations}.
They are used both on their own, and in
combination with \emph{dependency pair}
approaches~\cite{art:gie:00:1}.

In the higher-order world, monotonic algebras were among the first
termination methods to be defined, appearing as early as
1994~\cite{pol:94:1}; an in-depth study is done in van de Pol's 1996
PhD thesis~\cite{pol:96:1}.
Surprisingly,
the method has been almost entirely absent from the literature ever
since.  This is despite a lot of interest in higher-order rewriting,
witnessed not only by a fair number of publications,
but also by the recent participation of higher-order tools in the
\emph{annual Termination Competition} \cite{termcomp}.  Since the
addition of a higher-order category,
two tools have
participated: \textsf{THOR}~\cite{thor}, by
Borralleras and Rubio, and \wanda~\cite{wanda}, by the second author of this
paper.
So far,
neither tool has implemented weakly monotonic algebras.

In this paper we aim to counteract this situation, by both studying
the class of polynomial interpretations in the natural numbers, and
implementing the resulting technique in the termination tool \wanda.
Van de Pol did not consider automation of his method (there was less
focus on automation at the time),
but there are now years of experience of the first-order world to
build on; we will lift the parametric first-order approach
\cite{con:mar:tom:urb:05},
and make some necessary
adaptations
to cater for the presence of higher-order variables.

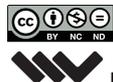
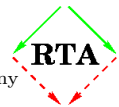
\paragraaf{Paper Setup} Section~\ref{sec:preliminaries} discusses
preliminaries:
\emph{Algebraic Functional Systems}, the higher-order formalism we
consider, reduction pairs
and
weakly monotonic
algebras for
typed $\lambda$-terms. 
In Section~\ref{sec:monalg} we extend these definitions to AFSs,
and define a general termination method.
Sec-\linebreak tion~\ref{sec:polynomial} defines the class of
\emph{higher-order
polynomials}, and in Section~\ref{sec:automation} we
show how suitable polynomial interpretations can be found automatically.
Experiments with this implementation are presented in
Section~\ref{sec:experiments}, and an overview
and ideas for future work are
given in Section~\ref{sec:conclusion}.

The main contribution of this paper are the techniques for
automation, discussed in Section~\ref{sec:experiments}.  For
simplicity of the code,
these techniques are limited to the (very common) class
of second-order AFSs, although extensions to systems of a higher
order are possible.  As far as we know, this is the first
implementation of higher-order polynomial interpretations.

\confreport{\emph{An extended version of this paper with more elaborate
proofs is available at~\cite{fuh:kop:12:2}.}}{\emph{This is a
pre-editing version of~\cite{fuh:kop:12:1}, including
some proof extensions in the appendix.}}

\section{Background}\label{sec:preliminaries}

\subsection{Algebraic Functional Systems}\label{subsec:afs}

We consider algebraic higher-order rewriting as defined by
Jouannaud and Okada, also called \emph{Algebraic Functional Systems (AFSs)} \cite{jou:oka:91:1}.
This formalism combines the simply-typed $\lambda$-calcu-\linebreak lus with algebraic
reduction, and appears in papers on e.g.\ HORPO~\cite{jou:rub:07:1},
MHOSPO~\cite{bor:rub:01:1} and dependency pairs~\cite{kop:raa:11:1};
it is also the formalism in the higher-order category of the
annual ter-\linebreak mination competition.
We follow roughly the definitions in~\cite[Ch.~11.2.3]{ter:03}, as
recalled below.

\paragraaf{Types and Terms}
The set of \emph{simple types} (or just \emph{types}) 
is generated from a given set $\setsorts$ of \emph{base types} and
the binary, right-associative type constructor $\typepijl$; types
are denoted by $\atype,\btype$ and base types by $\abasetype,
\bbasetype$.
A type with at least one occurrence of $\ftypepijl$
is called a \emph{functional type}.
A \emph{type declaration}
is an expression
of the form
$[\atype_1 \times \ldots \times \atype_n] \decpijl \btype$ for types
$\sigma_i,\tau$;
if $n = 0$
we just write $\btype$.
Type declarations are not types, but are used to
``type'' function symbols.
All types can be expressed in the form $\atype_1 \typepijl \ldots
\typepijl \atype_n \typepijl \abasetype$ (with $n \geq 0$ and
$\abasetype \in \setsorts$).
The \emph{order} of a type is
$\typedepth(\abasetype) = 0$ if $\abasetype \in \setsorts$, and
$\typedepth(\atype \typepijl \btype) =
\max(\typedepth(\atype)+1,\typedepth(\btype))$.
Extending this
to type declarations,
$\typedepth([\atype_1 \times \ldots \times \atype_n] \decpijl \btype) =
\max(\typedepth(\atype_1)+1,\ldots,\typedepth(\atype_n)+1,
\typedepth(\btype))$.

\smallskip
We assume a set $\setvar$
of infinitely many typed
variables for each type, and a set $\F$ disjoint from $\setvar$ which
consists of function symbols, each equipped with a type declaration.
\emph{Terms} over $\F$ are those expressions $\aterm$
for which we can infer $\aterm : \atype$ for some type $\atype$ using
the clauses:

\begin{tabular}{lllll}
  \clausevar &
  & $\avar : \atype$
  & \mbox{} & if $\avar:\atype \in \setvar$ \\
  \clauseapp &
  & $\aterm \cdot \bterm : \btype$
  & \mbox{}& if $\aterm : \atype \typepijl \btype$ and
      $\bterm : \atype$  \\
  \clauseabs & 
  & $\abs{\avar}{\aterm} : \atype \typepijl \btype$
  & \mbox{} & if $\avar:\atype \in \setvar$ and $\aterm : \btype$ \\
  \clausefun &
  & $\afun (\aterm_1 , \ldots , \aterm_n) : \btype$
  & \mbox{} &
  if $\afun : {[\atype_1 \times \ldots \times \atype_n] \typepijl
    \btype} \in \F$ and 
  $\aterm_1 :\atype_1 , \ldots, \aterm_n:\atype_n$
\end{tabular}

\noindent
Note that a function symbol $\afun : {[\atype_1 \times \ldots \times
\atype_n] \typepijl \btype}$ takes exactly $n$ arguments, and
$\btype$ is not necessarily a base type (a type declaration
gives the \emph{arity} of the symbol).
$\lambda$ binds occurrences of variables as in the
$\lambda$-calculus.
Terms are considered modulo $\alpha$-conversion; bound variables are
renamed if necessary.
Variables which are not bound are called \emph{free}, and
the set of free variables of $\aterm$ is denoted
$\FV(\aterm)$.
Application is left-associative, so
$\app{\app{\aterm}{\bterm}}{\cterm}$ should be read \linebreak
$\app{(\app{\aterm}{\bterm})}{\cterm}$.
Terms constructed without clause \clausefun\ are also called
\emph{(simply-typed) $\lambda$-terms}.

A \emph{substitution} $[\vec{\avar}:=\vec{\aterm}]$, 
with $\vec{\avar}$ and $\vec{\aterm}$
finite vectors of equal length, 
is the homomorphic extension of the type-preserving mapping
$\vec{\avar} \mapsto \vec{\aterm}$ from variables to terms.  
Substitutions are denoted $\asub, \bsub$,
and the result of applying $\asub$ to a term $\aterm$ 
is denoted $\subst{\aterm}{\asub}$.
The \emph{domain} $\domain(\gamma)$ of $\gamma =
[\vec{\avar}:= {\vec{\aterm}}]$ is $\{\vec{\avar}\}$.
Substituting does not bind free variables.
A \emph{context} $C[]$ is a term with a single occurrence of a
special symbol
$\Box_\atype$.
The result of replacing $\Box_\atype$ in $C[]$
by a term $\aterm$ of type $\atype$ 
is denoted $C[s]$.
Free variables may be captured;
if $C[] = \abs{\avar}{\Box_\atype}$ then $C[\avar] =
\abs{\avar}{\avar}$.

\pagebreak
\sparagraaf{Rules and Rewriting}
A \emph{rewrite rule} is a pair of terms $l \arrz r$ such that $l$
and $r$ have the same type and all free variables of $r$ also occur
in $l$.  In~\cite{kop:11:1} some termination-preserving
transformations on the general format of AFS-rules are presented;
using these results, we may additionally assume that $l$ has
the form $\afun(l_1,\ldots,l_n) \cdot l_{n+1} \cdots l_m$ (with $\afun
\in \F$ and $m
\geq n \geq 0$), that $l$ has no subterms $\app{\avar}{\aterm}$ with
$\avar$ a free variable, and that neither $l$ nor $r$ have a subterm
$(\abs{\avar}{\aterm}) \cdot \bterm$.
Given a set of rules $\Rules$, the \emph{rewrite} or \emph{reduction
relation} $\arr{\Rules}$ on terms is given by the following clauses:
\begin{tabular}{crcll}
\clauserule & $C[\subst{l}{\asub}]$ & $\arr{\Rules}$ &
  $C[\subst{r}{\asub}]$ & 
  with $l \arrz r \in \Rules$, $C$ a context,
  $\asub$ a substitution \\
\clausebeta & $C[(\abs{\avar}{\aterm}) \cdot \bterm]$ &
  $\arr{\Rules}$ & $C[\aterm[\avar:=\bterm]]$ &
  with $\aterm,\bterm$ terms, $C$ a context \\
\end{tabular} \\
An \emph{algebraic functional system (AFS)} is the combination of a
set of terms and a rewrite relation on this set, and is usually
specified by a pair $(\F,\Rules)$, or just by a set $\Rules$ of rules.  An AFS is
terminating if there is no infinite reduction $\aterm_1 \arr{\Rules}
\aterm_2 \arr{\Rules} \ldots$

An AFS is \emph{second-order} if
the type declarations of all function symbols have order $\leq 2$.
In a second-order system, all free variables in the rules have order
$\leq 1$ (this follows by the restrictions on the left-hand side),
and all bound variables have base type (this holds because
free variables have order $\leq 1$ and
we have assumed that the rules do not contain
$\beta$-redexes).

\begin{example} \label{ex:hshuffledef}
One of the examples considered in this paper is the AFS $\hshuffle$.
This (second-order) system for list manipulation has five
function symbols,
$\nil : \natlist,\ \cons : [\nat \times \natlist] \decpijl \natlist,\ 
\append : [\natlist \times \natlist] \decpijl \natlist,\ \reverse :
[\natlist] \decpijl \natlist,\ \hshuffle : [(\nat \typepijl \nat)
\times \natlist] \decpijl \natlist$, and the following
rules: \vspace{4px} \\
\indent
$
\begin{array}{rclrcl}
\append(\cons(h,t),l) & \arrz & \cons(h,\append(t,l)) &
\append(\nil,l) & \arrz & l \\
\reverse(\cons(h,t)) & \arrz & \append(\reverse(t),\cons(h,\nil)) &
\reverse(\nil) & \arrz & \nil \\
\hshuffle(F,\cons(h,t)) & \arrz & \cons(\app{F}{h},\hshuffle(F,
  \reverse(t)))\ \  &
\hshuffle(F,\nil) & \arrz & \nil \\
\end{array}
$
\end{example}

\subsection{Reduction Pairs}\label{subsec:proveterm}

To prove termination, modern approaches typically use
\emph{reduction pairs}, in one of three setups:

For \emph{rule removal}, we consider a \emph{strong reduction pair}:
a pair $(\geqterm,\gterm)$ of a quasi-ordering and a well-founded
ordering on terms, such that
  $\geqterm$ and $\gterm$ are \emph{compatible}:
  $\geqterm \cdot \gterm$ is included in $\gterm$ or
  $\gterm \cdot \geqterm$ is,
  both $\geqterm$ and $\gterm$ are \emph{monotonic},
  both $\geqterm$ and $\gterm$ are \emph{stable}
  (preserved under substitution),
  and in the higher-order case, \emph{$\geqterm$ contains $\beta$}:
  $\app{(\abs{\avar}{\aterm})}{\bterm} \geqterm \aterm[\avar:=\bterm]$.

If $\Rules = \Rules_1 \uplus \Rules_2$ and $l \gterm r$ for rules in
$\Rules_1$, and $l \geqterm r$ for rules in $\Rules_2$, then there
is no $\arr{\Rules}$-sequence which uses the rules in $\Rules_1$
infinitely often; this would contradict well-foundedness of $\gterm$.
Thus, $\arr{\Rules}$ is terminating
if $\arr{\Rules_2}$ is terminating.  In practice, we try
to orient all rules with either $\gterm$ or $\geqterm$, and then
remove those ordered with $\gterm$ and continue with the rest.

The second setup, \emph{dependency pairs}, is more sophisticated.
In this approach,
dependency pair chains are considered, which use infinitely many
``dependency pair'' steps at the top of a term.
It is
enough to orient the resulting constraints with a \emph{weak
reduction pair}: a pair $(\geqterm,\gterm)$ of a quasi-ordering and a
compatible well-founded ordering where both are stable, and
$\geqterm$ is monotonic and contains $\beta$.
The dependency pair approach was defined for first-order TRSs
in~\cite{art:gie:00:1}, and has seen many extensions and improvements
since.  For higher-order rewriting, two variations
exist: \emph{static dependency
pairs}~\cite{kus:iso:sak:bla:09:1} and
\emph{dynamic dependency
pairs}~\cite{sak:wat:sak:01,kop:raa:12:1}.

The static dependency pair approach is restricted to \emph{plain
function passing} systems; slightly simplified, whenever a
higher-order variable $F$ occurs in the right-hand side of a rule
$\afun(l_1,\ldots,l_n) \arrz r$, then $F$ is one of the $l_i$.  Static
dependency pairs may have variables in the right-hand side which do
not occur in the left (such as a dependency pair
$\up{\I}(\suc(n)) \dppijl \up{\I}(m)$), but always have the form
$\up{\afun}(l_1,\ldots,l_n) \dppijl \up{\bfun}(r_1,\ldots,r_m)$.  The
static approach gives constraints of the form $l \geqterm r$ or $l
\gterm r$ for dependency pairs $l \dppijl r$, and $l \geqterm r$ for
rules $l \arrz r$.

The dynamic dependency pair approach is unrestricted, but
right-hand sides of dependency pairs may be headed by a variable, e.g.
$\up{\collapse}(\cons(F,t)) \dppijl \app{F}{\collapse(t)}$, and
sometimes subterm steps are needed.  Thus, the dynamic
approach not only gives constraints
$l \gterm r$ or $l \geqterm r$ for dependency
pairs
and $l \geqterm r$ for rules,
but also
two further groups of constraints:
\begin{itemize}
\item $\app{\afun(\aterm_1,\ldots,\aterm_n)}{\bterm_1} \cdots
  \bterm_m \geqterm \app{\aterm_i}{\c_{\atype_1}} \cdots
  \c_{\atype_{k_i}}$ if both sides have base type, $\aterm_i : \atype_1
  \typepijl \ldots \typepijl \atype_{k_i} \typepijl \abasetype$ and
  $\afun$ is a symbol in some fixed set $S$ (the $\c_{\atype_j}$ are
  special symbols which may occur in the right-hand sides of
  dependency pairs but do not occur in the rules)
\item $\app{\aterm}{\bterm_1} \cdots \bterm_n \geqterm
  \app{\bterm_i}{\c_{\atype_1}} \cdots \c_{\atype_{k_i}}$ if both sides
  have base type, and $\bterm_i : \atype_1 \typepijl \ldots \typepijl
  \atype_{k_i} \typepijl \abasetype$
\end{itemize}

A third setup, which also uses a sort of reduction pair rather than
the traditional reduction ordering, are the \emph{monotonic semantic
path orderings} from Borralleras and Rubio~\cite{bor:rub:01:1}.  This
method is based on a
recursive path ordering, but uses a
well-founded order on terms rather than a precedence on function
symbols; this
gives constraints of the form $\aterm \succeq_I \bterm,\ \aterm
\succeq_Q \bterm,\ \aterm \succ_Q \bterm$, where $\succeq_I$ and
$\succeq_Q$ are quasi-orderings and
$\aterm \succeq_I \bterm$
implies $\afun(\ldots,\aterm,\ldots) \succeq_Q \afun(\ldots,\bterm,
\ldots)$.

\smallskip
In this paper, we focus on the first two setups, which have
been implemented in
\wanda.  However, the technique could be used with the monotonic
semantic path ordering as well.

\subsection{First-order Monotonic Algebras - Idea Sketch}
\label{subsec:fo}

In the first-order definition of monotonic
algebras~\cite{end:wal:zan:08}, terms are
mapped to elements of a
well-founded target domain $(\basealgebraset,>,\geq)$.
This is done by choosing an interpretation function
$\constvaluation(\afun)$ for all function symbols $\afun$ that is
monotonic w.r.t.\ $>$ and $\geq$,
and extending this homomorphically to an
interpretation $\algint{\cdot}$ of terms;
for
\emph{polynomial interpretations}, $\constvaluation(\afun)$ is always
a polynomial.
If $\algintc{l} > \algintc{r}$ for all valuations $\varvaluation$
of the free variables of $l$, then
$\algint{C[l\asub]}_\constvaluation >
\algint{C[r\asub]}_\constvaluation$ for all
contexts $C$
and
substitutions $\asub$.  Thus,
the pair $(\geqterm,\gterm)$ where $\aterm \geqterm \bterm$ if
$\algintc{\aterm} \geq \algintc{\bterm}$ and $\aterm \gterm \bterm$
if $\algintc{\aterm} > \algintc{\bterm}$
can be used as a strong reduction pair.

For example,
to prove termination of the TRS consisting of the two $\append$
rules from Example~\ref{ex:hshuffledef}, we might assign the
following interpretation to the function symbols:
$\constvaluation(\nil) = 2,\ \constvaluation(\cons) = \fatlambda n m.
n+m+1$ and $\constvaluation(\append) = \fatlambda n m.2 \cdot n + m
+ 1$.
Here, the $\fatlambda$ syntax indicates function creation: $\cons$,
for instance,
is mapped to a function which takes two arguments, and
returns their sum plus one.
Calculating all $\algintc{l},\algintc{r}$,
and noting that $(\N,>,\geq)$ is a
well-founded set and that all interpretations are
monotonic
functions, we see that the TRS is terminating because for all $h,t,l$:
$4\!+\!l\!+\!1 > l$ (for the rule $\append(\nil,l) \arrz l$), and
$2\!\cdot\!h\!+\!2\!\cdot\!t\!+\!2\!+\! l\!+\!1 >
h\!+\!2\!\cdot\!t\!+\!l\!+\!1$ (for $\append(\cons(h,t),l) \arrz
\cons(h,\append(t,l))$.

\subsection{Weakly Monotonic Functionals}\label{subsec:wmf}

In higher-order rewriting we have to deal with infinitely many types
(due to the type constructor $\typepijl$), a complication
not present in first-order rewriting.  As a consequence, it is not
practical to map all terms to
the same target
set.
A more natural interpretation would be, for instance, to
map a
functional term $\abs{\avar}{\aterm} : \o \typepijl \o$ to an element
of the function space $\N \functionpijl \N$.  However, this choice
has problems of its own, since it forces the termination prover to
deal with functions that absolutely nothing is known about.
Instead, the target domain for interpreting terms, as proposed by van
de Pol in~\cite{pol:96:1}, is the class of \emph{weakly monotonic
functionals}.
To each type $\atype$ we assign a set $\WM_\atype$ and
two relations: a well-founded ordering $\gwm_\atype$ and a 
quasi-ordering $\geqwm_\atype$.  Intuitively, the elements of
$\WM_{\atype \ftypepijl \btype}$ are functions which preserve
$\geqwm$.

\begin{definition}[Weakly Monotonic Functionals]
\cite[Def.\ 4.1.1]{pol:96:1}
We assume given a \emph{well-founded set}: a triple $\basealgebra =
(\basealgebraset, >, \geq)$ of a non-empty set, a well-founded
partial ordering on that set and a compatible
quasi-ordering.\footnote{Van de Pol defines $\geq$
as the reflexive closure of $>$.  In contrast, here we generalise the
notion of a well-founded set to include an explicitly given
compatible quasi-ordering $\geq$.
}
To each type $\atype$ we associate a set $\WM_\atype$ of
\emph{weakly monotonic functionals of type $\atype$}
and two relations $\gwm_\atype$ and $\geqwm_\atype$, defined
inductively as follows:

For a base type $\abasetype$, we have $\WM_\abasetype =
\basealgebraset$;\ $\gwm_\abasetype \mathord{=} >$, and
$\geqwm_\abasetype \mathord{=} \geq$.

For a functional type $\atype \typepijl \btype$,
  $\WM_{\atype \ftypepijl \btype}$ consists of the functions $f$ from
  $\WM_\atype$ to $\WM_\btype$
  such that:
  if $x \geqwm_\atype y$ then $f(x) \geqwm_\btype f(y)$.
  Let
  $f \gwm_{\atype \ftypepijl \btype} g$ iff $f(x) \gwm_\btype g(x)$,
  and
$f \geqwm_{\atype \ftypepijl \btype} g$ iff $f(x) \geqwm_\btype
  g(x)$
  for all $x,y \in \WM_\atype$.
\end{definition}

\noindent
Thus, $\WM_{\atype \ftypepijl \btype}$ is a subset of the function
space $\WM_\atype \functionpijl \WM_\btype$, consisting of functions
which preserve
$\geqwm$.
Note that both $\WM_\atype$ and
the relations $\gwm_\atype$ and $\geqwm_\atype$ should be
considered as
parametrised with $\basealgebra$; the complete notation would be
$(\WM_\atype^\basealgebra,\gwm_\atype^\basealgebra,
\geqwm_\atype^\basealgebra)$.
For readability, $\basealgebra$ will normally be omitted, as will the
type denotations for the
various $\gwm_\atype$ and $\geqwm_\atype$ relations.
The phrase ``$f$ is weakly monotonic'' means that $f \in \WM_\atype$
for some $\atype$.

It is not hard to see that an element $\fatlambda \avar_1 \ldots
\avar_n.P(\avar_1,\ldots,\avar_n)$ of the function space
$\WM_{\atype_1} \functionpijl \ldots \functionpijl \WM_{\atype_n}
\functionpijl \basealgebraset$ is weakly monotonic if and only if:
\[
\indent\indent\begin{array}{c}
\forall N_1,M_1 \in \WM_{\atype_1}, \ldots, N_n,M_n \in
\WM_{\atype_n}: \\
\mathrm{if\ each}\ N_i \geqwm M_i\ \mathrm{then}\ 
P(N_1,\ldots,N_n) \geqwm P(M_1,\ldots,M_n)
\end{array}
\]

\noindent
By Lemmas~4.1.3 and~4.1.4 in~\cite{pol:96:1} we obtain several
pleasant properties of $\geqwm$ and $\gwm$:

\begin{lemma}\label{lem:gwminteract}
For all types $\atype$,\ the relations $\gwm_\atype$ and
$\geqwm_\atype$ are compatible, $\gwm_\atype$ is well founded,
$\geqwm_\atype$ is reflexive, and
both $\gwm_\atype$ and $\geqwm_\atype$ are transitive.
\end{lemma}

\commentaar{the definition in~\cite{pol:96:1} actually
assigns a different
set $\basealgebra_\abasetype$ to
each base type $\abasetype$ (although there must be an addition
operator $+_{\abasetype,\bbasetype,\abasetype}$ for every
pair of base
types).
We
use the
same set for all base types, as this gives a
simpler
definition, and it is not obvious that using different sets gives
a
stronger technique; we could for instance choose $\basealgebra =
\basealgebra_\abasetype \uplus \basealgebra_\bbasetype$ instead.

Also, in~\cite{pol:96:1} $\WM_{\atype \ftypepijl \btype}$ consists of
functions $f$ in a larger function space
$\mathcal{I}_\atype
\functionpijl \mathcal{I}_\btype$\footnote{Here,
$\mathcal{I}_\abasetype = \basealgebra_\abasetype$ if $\abasetype \in
\setsorts$, and $\mathcal{I}_{\atype \ftypepijl \btype}$ is the full
function space $\mathcal{I}_\atype \functionpijl \mathcal{I}_\btype$.}
such that $f(x) \in \WM_\btype$ if $x \in \WM_\atype$ and $f$
preserves $\geqwm$.
Our definition is essentially equivalent;
every function in
$\WM_{\atype} \functionpijl \WM_{\btype}$ can be extended to a
function in $\mathcal{I}_\atype \functionpijl \mathcal{I}_\btype$.
}

\begin{example}[Some Examples of Weakly Monotonic Functionals]
\label{exa:wmf}\
\begin{enumerate}
\item\label{ex:wmf:constants}
  \emph{Constant Function:}
  For all $n \in \basealgebraset$ and types $\btype =
  \btype_1 \typepijl \ldots \typepijl \btype_k \typepijl \abasetype$,
  let $n_\btype := \fatlambda \vec{x}.n$.
  Then $n_\btype \in \WM_\btype$, since $n_\btype(N_1,\ldots,N_k)
  = n \geqwm n = n_\btype(M_1,\ldots,M_k)$ if all $N_i \geqwm M_i$.
\item\label{ex:wmf:lowval}
  \emph{Lowest Value Function:}
  Suppose $\basealgebraset$ has
  a
  minimal element $0$ for the ordering
  $>$.  Then for any type $\btype = \btype_1 \typepijl \ldots
  \typepijl \btype_k \typepijl \abasetype$ the function $\fatlambda
  \afun.\afun(\vec{0})$, which maps $\afun \in \WM_\btype$ to
  $\afun(0_{\btype_1},\ldots,0_{\btype_k})
  $ (where each $0_{\btype_i}$ is a constant function),
  is 
  in $\WM_{\btype \ftypepijl \o}$
  by induction on $k$.
\item\label{ex:wmf:max}
  \emph{Maximum Function:}
  In the natural numbers, the function $\max$ which assigns to
  any two numbers the highest of the two is weakly monotonic, since
  $\max(a,b) \geq \max(a',b')$ if $a \geq a'$ and $b \geq b'$.
  For any type $\btype = \btype_1 \typepijl \ldots \typepijl \btype_k
  \typepijl \abasetype$ (with $\abasetype \in \setsorts$) let
  $\max_\btype(f,m) = \fatlambda \avar_1 \ldots \avar_k.
  \max(f(\avar_1,\ldots,\avar_k),m)$.
  This function is in $\WM_{\btype \ftypepijl \abasetype \ftypepijl
  \btype}$ by induction on $k$.
\end{enumerate}
The constant and lowest value function appear in~\cite{pol:96:1};
the maximum function appears in~\cite{kop:raa:11:1}.
\end{example}

\begin{definition}[Interpreting a $\lambda$-Term to a Weakly
Monotonic Functional]\label{def:alginthrs}

Given a well-founded set $\basealgebra = (\basealgebraset,>,\geq)$, a
simply-typed $\lambda$-term $\aterm$ and a \emph{valuation} 
$\varvaluation$
which assigns to all variables $\avar : \atype$ in $\FV(\aterm)$ an
element of $\WM_\atype$, let $\lamalgint{\aterm}_\varvaluation$ be
defined by the following
clauses:
\[
\begin{array}{llll}
\lamalgint{\avar}_\varvaluation & = & \varvaluation(\avar) &
  \mathrm{if}\ \avar \in \setvar \\
\lamalgint{\app{\aterm}{\bterm}}_\varvaluation & = &
  \lamalgint{\aterm}_\varvaluation(\lamalgint{\bterm}_\varvaluation)
  \\
\lamalgint{\abs{\avar}{\aterm}}_\varvaluation & = & \fatlambda n.
  \lamalgint{\aterm}_{\varvaluation \cup \{\avar \mapsto
  n \}} & \mathrm{if}\ \avar \notin \domain(\varvaluation)\ \ 
  \text{(always applicable with $\alpha$-conversion)} \\
\end{array}
\]
\end{definition}

Definition~\ref{def:alginthrs} is an instance of a definition
in~\cite{pol:96:1} which suffices for the extension to AFSs.
By Lemma~3.2.1 and Proposition~4.1.5(1) in~\cite{pol:96:1}, we have:

\begin{lemma}[Facts on $\lambda$-Term Interpretations]
\label{lem:algintfacts}
\
\begin{enumerate}
\item \emph{(Substitution Lemma)}
  \label{lem:algintfacts:substitution}
  Given a substitution $\asub = [\avar_1:=\aterm_1,\ldots,\avar_n:=
  \aterm_n]$ and a valuation $\varvaluation$ whose domain does not
  include the $\avar_i$: $\lamalgint{\aterm\asub}_\varvaluation =
  \lamalgint{\aterm}_{\varvaluation \circ \asub}$.
  Here, $\varvaluation \circ \asub$ is the valuation $\varvaluation
  \cup \{ \avar_1 \mapsto \lamalgint{\aterm_1}_\varvaluation, \ldots,
  \avar_n \mapsto \lamalgint{\aterm_n}_\varvaluation \}$.
\item
  \label{lem:algintfacts:interprete}
  If $\aterm : \atype$ is a simply-typed $\lambda$-term, then
  $\lamalgint{\aterm}_\varvaluation \in \WM_\atype$ for all
  valuations $\varvaluation$.
\end{enumerate}
\end{lemma}

\section{(Weakly and Extended) Monotonic Algebras for AFSs}\label{sec:monalg}

The theory in~\cite{pol:96:1} was defined for Nipkow's formalism of
\emph{Higher-order Rewrite Systems (HRSs)}~\cite{nip:91:1}, which
differs in several
ways from our \emph{Algebraic Functional Systems}.  Most importantly,
in the setting of HRSs
terms are equivalence classes modulo $\beta$;
thus, the definitions in~\cite{pol:96:1} are designed
so
that
$\algint{\aterm} = \algint{\bterm}$ if $\aterm$ and $\bterm$ are
equal modulo $\beta$.
This is not convenient for AFSs,
since then
for instance
$\algint{\app{(\abs{\avar}{\nul})}{\bterm}} =
\algint{\app{(\abs{\avar}{\nul})}{\cterm}}$ regardless of
$\bterm$ and $\cterm$.

Fortunately, we do not need to redesign the whole theory for use
with AFSs; rather, we can transpose the result using a
transformation.  We will need no more than
Lemma~\ref{lem:algintfacts}.

\note{some of the results of this section have also been
stated
in~\cite{kop:raa:11:1}, but the results there
are limited to what is needed for the dynamic dependency pair
approach; here, we are more general, by not fixing the interpretation of
application and
also
studying strong
monotonicity.}

\begin{definition}[Weakly Monotonic Algebras for AFSs]
\label{def:wmf:afs}
A \emph{weakly monotonic algebra} for an AFS with function symbols
$\F$ consists of a well-founded set $\basealgebra =
(\basealgebraset,>,\geq)$ and an
\emph{interpretation function} $\constvaluation$ which assigns
an element of $\WM_{\atype_1 \ftypepijl \ldots \ftypepijl
\atype_n \ftypepijl \btype}$
to all $\afun : [\atype_1 \times \ldots \times \atype_n] \decpijl
\btype \in \F$, and
a value in
$\WM_{\atype \ftypepijl \atype}$ to the fresh symbol $@^\atype$ for
all functional types $\atype$.

Given an algebra $(\basealgebra,\constvaluation)$, a term $\aterm$
over $\F$ and a \emph{valuation} $\alpha$ which assigns to all
variables $\avar : \atype$ in $\FV(\aterm)$ an element of
$\WM_\atype$, let $\algintc{\aterm}$ be defined recursively as
follows: \\\indent
$
\begin{array}{llll}
\algintc{\avar} & = & \varvaluation(\avar) & \mathrm{if}\ \avar
  \in \setvar \\
\algintc{\afun(\aterm_1,\ldots,\aterm_n)} & = &
  \constvaluation(\afun)(\algintc{\aterm_1},\ldots,\algintc{\aterm_n})
  & \mathrm{if}\ \afun \in \F \\
\algintc{\app{\aterm}{\bterm}} & = & \constvaluation(@^\atype)(
  \algintc{\aterm},\algintc{\bterm}) & \mathrm{if}\ \aterm : \atype \\
\algintc{\abs{\avar}{\aterm}} & = & \fatlambda n.
  \algint{\aterm}_{\constvaluation,\varvaluation \cup \{\avar \mapsto
  n \}} & \mathrm{if}\ \avar \notin \domain(\varvaluation) \\
\end{array}
$
\end{definition}

\noindent
This definition, which roughly follows the ideas
of~\cite{pol:96:1} and extends the definition of a weakly
monotone algebra in~\cite{end:wal:zan:08} to the setting of AFSs,
assigns to every function symbol and
variable a weakly monotonic functional, and calculates the value of
the term accordingly.  For the purposes of the interpretation,
application is treated as a function symbol $@^\atype$.
As in~\cite{pol:96:1}, the interpretation function
$\constvaluation$ is separate from the valuation $\varvaluation$,
as we will quantify over $\varvaluation$.

\begin{example}
Consider the $\hshuffle$ signature from Example~\ref{ex:hshuffledef},
extended with symbols $\nul$ and $\suc$ for
the natural
numbers.  Let $\basealgebra = (\N,>,\geq)$.
By way of example, choose:
$
\constvaluation(\nul) = 1,\ 
\constvaluation(\suc) = \fatlambda n.n+2,\ 
\constvaluation(\cons) = \fatlambda n m.n+m,\ 
\constvaluation(\hshuffle) = \fatlambda F n.F(n)$
and
$\varvaluation(\cvar) = 37$.
Then $\algintc{\hshuffle(\abs{\avar}{\suc(\avar)},\cons(\suc(\nul),
\cvar))} = \algint{F(n)}_{\constvaluation,\{F \mapsto \fatlambda m.
m+2, n \mapsto 40\}} = 42$.
\end{example}

\begin{lemma}[Weakly Monotonic Algebras for AFSs]\label{lem:wmf:afs}
Let $(\basealgebra,\constvaluation)$ be a weakly monotonic algebra
for $\F$, and $\aterm,\bterm$ terms over $\F$.
For all valuations $\varvaluation$ as described in
Definition~\ref{def:wmf:afs}:
\begin{enumerate}
\item \label{lem:wmf:afsm:wm}
  $\algintc{\aterm} \in \WM_\atype$ if $\aterm : \atype$.
\item \label{lem:wmf:afsm:substitution}
  $\algint{\aterm}_{\constvaluation,\varvaluation \circ \asub} =
  \algintc{\aterm\asub}$
  (where $\varvaluation \circ \asub = \varvaluation \cup \{
  \avar \mapsto \algintc{\asub(\avar)} \mid \avar \in \domain(\asub)
  \}$)
\item \label{lem:wmf:afsm:stable}
  If $\algint{\aterm}_{\constvaluation,\varvaluationb} \geqwm
  \algint{\bterm}_{\constvaluation,\varvaluationb}$ for all
  valuations $\varvaluationb$, then $\algintc{\aterm\asub}
  \geqwm \algintc{\bterm\asub}$. \\
  If $\algint{\aterm}_{\constvaluation,\varvaluationb} \gwm
  \algint{\bterm}_{\constvaluation,\varvaluationb}$ for all
  valuations $\varvaluationb$, then $\algintc{\aterm\asub} \gwm
  \algintc{\bterm\asub}$.
\item \label{lem:wmf:afsm:context}
  If $\algint{\aterm}_{\constvaluation,\varvaluationb} \geqwm
  \algint{\bterm}_{\constvaluation,\varvaluationb}$ for all
  valuations $\varvaluationb$, then
  $\algintc{C[\aterm]} \geqwm \algintc{C[\bterm]}$.
\end{enumerate}
\end{lemma}

\begin{proof}
The proof proceeds by translating (arbitrary) terms to simply-typed
$\lambda$-terms, and then reusing the original result.  Interpretation
of function symbols ($\constvaluation$) is translated to assignment of
variables ($\varvaluation$), and application is treated as a function
symbol.

Consider the following transformation:
\[
\begin{array}{rclrcll}
\varphi(\avar) & = & \avar\ \ (\avar \in \setvar)\ \ \ \  &
\varphi(\afun(\aterm_1,\ldots,\aterm_n)) & = & \app{\avar_\afun}{
  \varphi(\aterm_1)} \cdots \varphi(\aterm_n) & (\afun \in \F) \\
\varphi(\abs{\avar}{\aterm}) & = & \abs{\avar}{\varphi(\aterm)} &
\varphi(\app{\aterm}{\bterm}) & = & \app{\app{\avar_{@,\atype}}{
  \varphi(\aterm)}}{\varphi(\bterm)} & (\aterm : \atype) \\
\end{array}
\]
Here, the $\avar_\afun$ is a new variable of type
$\atype_1 \typepijl \ldots \typepijl \atype_n \typepijl \btype$
for $\afun : [\atype_1 \times \ldots \times \atype_n] \decpijl
\btype \in \F$, and $\avar_{@,\atype}$ is a variable of type
$\atype \typepijl \atype$.
For any substitution $\asub$, let $\asub^\varphi$ denote the
substitution
$[\avar:=\varphi(\asub(\avar)) \mid \avar \in \domain(\asub)]$
(the $\avar_\afun$ are left
alone).  We make the following observations:

\emph{(**) $\varphi(\aterm\asub) = \varphi(\aterm)\asub^\varphi$ for
all substitutions $\asub$}.

\emph{(***) $\algintc{\aterm} = \lamalgint{\varphi(\aterm)}_{
\varvaluationb}$, if $\varvaluationb(\avar) = \varvaluation(\avar)$
for $\avar \in \FV(\aterm)$,
$\varvaluationb(\avar_\afun) =
\constvaluation(\afun)$,\ $\varvaluationb(\avar_{@,\atype}) =
\constvaluation(@^\atype)$}

Both statements hold by a straightforward induction on the form of
$\aterm$.

\smallskip\noindent
\textbf{(\ref{lem:wmf:afsm:wm})} holds by (***) and
Lemma~\ref{lem:algintfacts}(\ref{lem:algintfacts:interprete}).
\textbf{(\ref{lem:wmf:afsm:substitution})} holds because
$\algintc{\aterm\asub} =
\lamalgint{\varphi(\aterm \asub)}_\varvaluationb$ by (***), $=
\lamalgint{\varphi(\aterm)\asub^\varphi}_\varvaluationb$ by (**),
$= \lamalgint{\varphi(\aterm)}_{\varvaluationb \circ \asub^\varphi}$
by Lemma~\ref{lem:algintfacts}(\ref{lem:algintfacts:substitution}),
which is exactly $\algint{\aterm}_{\constvaluation,\varvaluation
\circ \asub}$ by (***).
\textbf{(\ref{lem:wmf:afsm:stable})} holds by
(\ref{lem:wmf:afsm:substitution}):
$\algintc{\aterm \asub} =
\algint{\aterm}_{\constvaluation,\varvaluation \circ \asub}$ by
(\ref{lem:wmf:afsm:substitution}),
$\geqwm \algint{\bterm}_{\constvaluation,\varvaluation \circ \asub}
=
\algintc{\bterm\asub}$, and similar for $\gwm$.
\textbf{(\ref{lem:wmf:afsm:context})} holds by a straightforward
induction on the form of $C$ (this result has no counterpart
in~\cite{pol:96:1}).
\end{proof}

\noindent
The theory so far allows us to use weakly monotonic algebras in a
\emph{weak reduction pair}.

\begin{theorem}\label{thm:weakredpair}
Let a weakly monotonic algebra $(\basealgebra,\constvaluation)$ be
given such that always $\constvaluation(@^\atype) \geqwm \fatlambda
f n.f(n)$, and define the pair $(\geqterm,\gterm)$ by:
$\aterm \geqterm \bterm$ if $\algintc{\aterm} \geqwm \algintc{\bterm}$
for all valuations $\varvaluation$, and $\aterm \gterm \bterm$ if
$\algintc{\aterm} \gwm \algintc{\bterm}$ for all $\varvaluation$.
Then $(\geqterm,\gterm)$ is a weak reduction pair.
\end{theorem}

\begin{proof}
$(\geqterm,\gterm)$ is a compatible combination of a quasi-ordering
and a well-founded ordering
by Lemma~\ref{lem:gwminteract} and $\geqterm$ is
monotonic by Lemma~\ref{lem:wmf:afs}(\ref{lem:wmf:afsm:context}).
Also, $\geqterm$ contains $\mathtt{beta}$: for all valuations
$\varvaluation$,\ $\algintc{\app{(\abs{\avar}{\aterm})}{\bterm}} =
\constvaluation(@^\atype)(\algintc{\abs{\avar}{\aterm}},
  \algintc{\bterm}) \geqwm
\algintc{\abs{\avar}{\aterm}}(\algintc{\bterm})$ by assumption,
which equals
$\algint{\aterm}_{\constvaluation,\varvaluation \cup \{\avar
  \mapsto \algintc{\bterm}\}} =
\algint{\aterm}_{\constvaluation,\varvaluation \circ [\avar:=
  \bterm]}$, and this equals
  $\algintc{\aterm[\avar:=\bterm]}$ by
Lemma~\ref{lem:wmf:afs}(\ref{lem:wmf:afsm:substitution}).
\end{proof}

\commentaar{if we choose
$\constvaluation(@^\atype) = \fatlambda f n.f(n)$, we have a system
very similar to the one used for simply-typed $\lambda$-calculus (and
HRSs).  By not fixing the interpretation of
$@^\atype$ we have a choice, which, depending on the setting
(rule removal, static dependency pairs, dynamic dependency pairs)
may be essential; we will see different choices in
Examples~\ref{ex:staticdp}, \ref{ex:dynamicdp}
and~\ref{ex:rulesremove}.
}

\begin{example}\label{ex:staticdp}
Using the static dependency pair framework
of~\cite{kus:iso:sak:bla:09:1} to deal with
$\hshuffle$, we obtain several sets of requirements.
HORPO~\cite{jou:rub:07:1}
runs into trouble with the dependency pair \newline
$\up{\hshuffle}(F,\cons(h,t)) \dppijl \up{\hshuffle}(F,\reverse(t))$,
where we need a weak reduction pair satisfying:
\[
\begin{array}{rclrcl}
\up{\hshuffle}(F,\cons(h,t)) & \gterm & \up{\hshuffle}(F,\reverse(t)) \\
\append(\cons(h,t),l) & \geqterm & \cons(h,\append(t,l)) &
\append(\nil,l) & \geqterm & l \\
\reverse(\cons(h,t)) & \geqterm & \append(\reverse(t),\cons(h,\nil)) &
\reverse(\nil) & \geqterm & \nil \\
\end{array}
\]
Using Theorem~\ref{thm:weakredpair}, we choose
the following interpretation $\constvaluation$ in the natural numbers: \\
\phantom{i}
$
\begin{array}{rclrclrcl}
\constvaluation(\up{\hshuffle}) & = & \fatlambda f n.n &
\constvaluation(\cons) & = & \fatlambda n m.m+1 &
\constvaluation(\nil) & = & 0 \\
\constvaluation(\reverse) & = & \fatlambda n.n &
\constvaluation(\append) & = & \fatlambda n m.n+m &
\constvaluation(@^\atype) & = & \fatlambda f n.f(n)\ 
  \mathrm{for\ all}\ \atype \\
\end{array}
$ \smallskip \\
Quantifying over the valuation, it suffices to show that for all $F
\in \WM_{\nat \ftypepijl \nat}, h,t \in \N$:
$t+1 > t,\ 
t+l+1 \geq t+l+1,\ 
t+1 \geq t+1,\ 
l \geq l,\ 
0 \geq 0$. This is obviously the case!
\end{example}

\begin{example}\label{ex:dynamicdp}
For a case where we cannot choose $\constvaluation(@^\atype) =
\fatlambda f n.f(n)$, consider
$\collapse$:
\[
\begin{array}{rclrclrcl}
\nul & : & \nat &
\minimum & : & [\nat \times \nat] \decpijl \nat &
\cons & : & [(\nat \typepijl \nat) \times \funlist] \decpijl \funlist \\
\suc & : & [\nat] \typepijl \nat &
\diff & : & [\nat \times \nat] \decpijl \nat &
\build & : & [\nat] \decpijl \funlist \\
\nil & : & \funlist &
\gcd & : & [\nat \times \nat] \decpijl \nat &
\collapse & : & [\funlist] \decpijl \nat \\
\end{array}
\]
\pagebreak
\[
\begin{array}{rclrcl}
\minimum(x,\nul) & \arrz & \nul &
\gcd(\suc(x),\nul) & \arrz & \suc(x) \\
\minimum(\nul,x) & \arrz & \nul &
\gcd(\nul,\suc(x)) & \arrz & \suc(x) \\
\minimum(\suc(x),\suc(y)) & \arrz & \suc(\minimum(x,y)) &
\gcd(\suc(x),\suc(y)) & \arrz & \gcd(\diff(x,y),\suc(\minimum(x,y)))\!\! \\
\diff(x,\nul) & \arrz & x &
\build(\nul) & \arrz & \nil \\
\diff(\nul,x) & \arrz & x &
\build(\suc(x)) & \arrz & \cons(\abs{y}{\gcd(y,x)},\build(x))\!\! \\
\diff(\suc(x),\suc(y)) & \arrz & \diff(x,y) &
\collapse(\nil) & \arrz & \nul \\
& & \multicolumn{2}{r}{
\collapse(\cons(F,t))} & \arrz & \app{F}{\collapse(t)} \\
\end{array}
\]
This AFS is not plain function passing, so we cannot use static
dependency pairs.
Using dynamic dependency pairs, HORPO
runs into trouble when faced with the
constraints:
\[
\begin{array}{rclrcl}
\up{\collapse}(\cons(F,t)) & \geqorgterm & \app{F}{\collapse(t)} \\
\up{\collapse}(\cons(F,t)) & \geqorgterm & \up{\collapse}(t)\ \  &
l & \geqterm & r\ \ \mathrm{for\ all\ rules}\ l \arrz r\ \mathrm{listed\ 
  above} \\
\end{array}
\]
The $\geqorgterm$ relation denotes that the constraint can either be
oriented with $\geqterm$ or with $\gterm$; to make progress, at least
one of these constraints must be oriented with $\gterm$.
Recall that in the dynamic dependency pair approach the constraints
must be satisfied with a reduction pair that also has
$\app{\aterm}{\bterm_1} \cdots \bterm_n \geqterm \app{\bterm_i}{c_1}
\cdots c_m$ if both sides have base type, for fresh constants $c_j$;
moreover, we must have $\gcd(\avar,\bvar) \geqterm \avar,\bvar$.
To guarantee this,
we choose $\constvaluation(@^{\atype
\ftypepijl \btype}) = \fatlambda f n.\max_\btype(f(n),n(\vec{0}))$,
where $n(\vec{0})$ and $\max_\btype$ were defined in
Example~\ref{exa:wmf}.
Then $\constvaluation(@^{\atype \ftypepijl \btype})
\geqwm \fatlambda f n.f(n)$, and if we assign $\constvaluation(c_j)
= 0_\atype$ for $c_j : \atype$, then
$\algintc{\app{\aterm}{\vec{\bterm}}} \geqwm
\algintc{\app{\bterm_i}{\vec{c}}}$ is indeed satisfied.
Additionally, let
$\constvaluation(\nul) = \constvaluation(\nil) = 0,\ 
\constvaluation(\diff) = \constvaluation(\gcd) =
  \fatlambda n m.n + m,\ 
\constvaluation(\suc) = \constvaluation(\build) =
  \fatlambda n.3 \cdot n,\ 
\constvaluation(\minimum) = \fatlambda n m.0,\ 
\constvaluation(\collapse) = \fatlambda n.n,\ 
\constvaluation(\up{\collapse}) = \fatlambda n.n + 1$ and
$\constvaluation(\cons) = \fatlambda f n.f(n) + n$.

With this interpretation, we have $l \geqterm r$ for all rules.
Moreover, $\algintc{\up{\collapse}(\cons(F,t))} =
1 + F(t) + t > \max(F(t),t) = \algintc{\app{F}{\collapse(t)}}$ and
$\algintc{\up{\collapse}(\cons(F,t))} = 1 + F(t) + t \geq
1 + t = \algintc{\up{\collapse}(t)}$.
As required, we can remove one dependency pair (the first one).
\end{example}

\sparagraaf{Strong Monotonicity}
To use weakly monotonic algebras in the setting of rule removal, we
shall need an additional requirement: $\gwm$ must be monotonic.  This
is achieved by posing a restriction on $\constvaluation$: each
$\constvaluation(\afun)$ should be \emph{strongly monotonic}:

\begin{definition}[Strongly Monotonic Functional]
An element $f$ of $\WM_{\atype_1 \ftypepijl \ldots \ftypepijl
\atype_n \ftypepijl \abasetype}$ is \emph{strongly monotonic in
argument $i$}
if for all $N_1 \in \WM_{\atype_1},\ldots,N_n \in \WM_{\atype_n}$ and
$M_i \in \WM_{\atype_i}$ we have:
$f(N_1,\ldots,N_i,\ldots,N_n) \gwm f(N_1,\ldots,M_i,\ldots,N_n)$
if $N_i \gwm M_i$.
\end{definition}
For first- and second-order functions, strong monotonicity
corresponds with the notion \emph{strict}
in~\cite{pol:96:1}.  For higher-order functions,
the definition
of~\cite{pol:96:1}
is more permissive.
We have chosen to use strong monotonicity
because the strictness requirement significantly complicates the
theory of~\cite{pol:96:1}, and
most common examples of
higher-order systems are second-order.
Strongly monotonic functionals exist for all types, e.g.\ 
$\fatlambda \avar_1 \ldots \avar_n.\avar_1(\vec{0}) + \ldots +
\avar_n(\vec{0}) \in \WM_{\btype_1 \ftypepijl \ldots \ftypepijl
\btype_n \ftypepijl \abasetype}$.

An \emph{extended monotonic algebra} is a weakly monotonic algebra
where each $\constvaluation(@^\atype)$ is strongly monotonic in its
first two arguments,\footnote{Note that e.g.\ 
$\constvaluation(@^{\o \ftypepijl \o \ftypepijl \o})$ is an element
of the function space $\WM_{\o \ftypepijl \o
\ftypepijl \o} \functionpijl \WM_\o \functionpijl \WM_\o
\functionpijl \WM_\o$; a function which takes \emph{three}
arguments.  It need not be strongly monotonic in
its 3$^{\mathrm{rd}}$ argument, because we think of application as a
symbol $@^{\atype \ftypepijl \btype} : [(\atype \typepijl \btype)
\times \atype] \decpijl \btype$ of arity \emph{2}, where $\btype$ may
be functional.} and for $\afun :
[\atype_1 \times \ldots \times \atype_n] \decpijl \btype \in \F$
also $\constvaluation(\afun)$ is strongly monotonic in its first $n$
arguments.
This notion
extends the
corresponding definition from \cite{end:wal:zan:08} for the first-order
setting to the setting of AFSs.
We obtain:

\begin{theorem}\label{thm:strongredpair}
Let an extended monotonic algebra $(\basealgebra,\constvaluation)$ be
given such that always $\constvaluation(@^\atype) \geqwm \fatlambda
f n.f(n)$; the pair $(\geqterm,\gterm)$ from
Theorem~\ref{thm:weakredpair} is a strong reduction pair.
\end{theorem}

\begin{proof}
It is a weak reduction pair by Theorem~\ref{thm:weakredpair}, and
strongly monotonic because
$\algintc{C[\aterm]} \gwm \algintc{C[\bterm]}$ for
all $\varvaluation$ whenever $\algintc{\aterm} \gwm
\algintc{\bterm}$ for all $\varvaluation$ (an easy induction).
\end{proof}

\section{Higher-Order Polynomial Interpretations}\label{sec:polynomial}

It remains to be seen how to \emph{find} suitable polynomial
interpretations, preferably automatically.  In this section, we will
discuss the class of \emph{higher-order polynomials over $\N$}, a
specific subclass of the weakly monotonic functionals with
$(\N,>,\geq)$ as a well-founded base set.  In the
following, we will see
how suitable polynomials can be found automatically.

\begin{definition}[Higher-Order Polynomial over $\N$]\label{def:hopol}
For a set $X = \{ \avar_1 : \atype_1,\ldots,\avar_n : \atype_n \}$ of
variables, each equipped with a type, the set $\Pol(X)$ of
\emph{higher-order polynomials} in $X$ is given by the following
clauses:
\begin{itemize}
\item if $n \in \N$, then $n \in \Pol(X)$;
\item if $p_1, p_2 \in \Pol(X)$, then $p_1 + p_2 \in \Pol(X)$
  and $p_1 \cdot p_2 \in \Pol(X)$;
\item if $\avar_i : \btype_1 \typepijl \ldots \typepijl \btype_m
  \typepijl \abasetype \in X$ with $\abasetype \in \setsorts$, and
  $p_1 \in \Pol^{\btype_1}(X),\ldots,p_m \in \Pol^{\btype_m}(X)$,
  then $\avar_i(p_1,\ldots,p_m) \in \Pol(X)$;
  \begin{itemize}
  \item
    here, $\Pol^\abasetype(X) = \Pol(X)$ for base types $\abasetype$,
    and $\Pol^{\atype \ftypepijl \btype}(X)$ contains functions
    $\fatlambda \bvar.p \in \WM_\atype$ with $p \in \Pol^\btype(X
    \cup \{\bvar\})$.
  \end{itemize}
\end{itemize}
\end{definition}

\noindent
We do not fix the set $X$.
A \emph{higher-order polynomial} is an element of any $\Pol(X)$.

Noting that $\WM_\atype = \WM_\btype$ if $\atype$ and $\btype$ have
the same ``form'' (so are equal modulo renaming of base types), the
following lemma holds for all $\abasetype \in \setsorts$:

\begin{lemma}\label{lem:polynomial:weak}
If $p \in \Pol(\{ \avar_1 : \atype_1, \ldots, \avar_n : \atype_n
\})$, then $\fatlambda \avar_1 \ldots \avar_n.p$
$\in \WM_{\atype_1 \ftypepijl \ldots \ftypepijl
\atype_n \ftypepijl \abasetype}$.
\end{lemma}

\begin{proof}
It is easy to see that $+$ and $\cdot$ are weakly monotonic.
Taking this into account, the lemma follows quickly with induction on
the
size of $p$, using
Lemma~\ref{lem:algintfacts}(\ref{lem:algintfacts:interprete}).
For the variable case,
if $\fatlambda \vec{\bvar}.p_i \in
\Pol^{\btype_i}(\{\vec{\avar}\})$, then $p_i \in \Pol(\{\vec{\avar},
\vec{\bvar}\})$, so the induction hypothesis applies.
\end{proof}

Higher-order polynomials are typically represented in the form
$a_1 + \ldots + a_n$ (with $n \geq 0$), where each $a_i$ is a
\emph{higher-order monomial}: an expression of the form
$b \cdot c_1 \cdots c_m$, where $b \in \N$ and each $c_i$ is either
a base-type variable $\avar$ or a function application
$\avar(\fatlambda \vec{\bvar_1}.p_1,\ldots,\fatlambda
\vec{\bvar_k}.p_k)$ with all $p_j$ higher-order polynomials again.
Examples of higher-order polynomials over the natural numbers are
for instance $0$ and $3+5 \cdot \avar^2 \cdot \bvar + F(37+\avar)$.
To find a \emph{strongly monotonic} functional, it suffices to
include, for all variables, a monomial containing only that variable:

\begin{lemma}\label{lem:strongpolynomial}
Let $P(\avar_1,\ldots,\avar_n)$ be a higher-order polynomial of the
form $p_1(\vec{\avar}) + \ldots + p_m(\vec{\avar})$, where all $p_i(
\vec{\avar})$ are higher-order monomials.  Then $\fatlambda \vec{
\avar}.P(\vec{\avar})$ is strongly monotonic in argument $i$ if there
is some $p_j$ of the form $a \cdot
\avar_i(\vec{b}(\vec{\avar}))$, where $a \in \N^+$.
\end{lemma}

\begin{proof}
Let $\avar_i \gwm \avar_i'$, so also $\avar_i \geqwm \avar_i'$ (since
$> \: \subseteq \: \geq$).
Let $\vec{\avar} := \avar_1,\ldots,\avar_i,\ldots,\avar_l$ and
$\vec{\avar'} := \avar_1,\ldots,\avar_i',\ldots,\avar_l$.
All $p_k$ are weakly monotonic by Lemma~\ref{lem:polynomial:weak},
so $p_k(\vec{\avar}) \geqwm p_k(\vec{\avar'})$.  Since $p_j(\vec{
\avar}) \gwm p_j(\vec{\avar'})$ and $+$ is strongly monotonic,
indeed $P(\vec{\avar}) \gwm P(\vec{\avar'})$.
\end{proof}

\begin{example}\label{ex:rulesremove}
For rule removal on the AFS $\hshuffle$ from
Ex.~\ref{ex:hshuffledef}, consider the interpretation: \\
\phantom{i}
$
\begin{array}{rclrcl}
\constvaluation(\append) & = & \fatlambda n m.n+m &
\constvaluation(\cons) & = & \fatlambda n m.n+m+3 \\
\constvaluation(\reverse) & = & \fatlambda n.n+1 &
\constvaluation(\nil) & = & 0 \\
\constvaluation(\hshuffle) & = & \fatlambda F n.2n+F(0)+nF(n)+1 &
\constvaluation(@^\atype) & = & \fatlambda f n \vec{m}.f(n,\vec{m}) +
  n(\vec{0})\ (**) \\
\end{array}
$
(**) Here, $n(\vec{0})$ is the ``lowest value'' function from
Ex.~\ref{exa:wmf}.  With this interpretation, which is a
strongly monotonic polynomial interpretation by
Lemma~\ref{lem:strongpolynomial}, all rules are oriented with
$\geqterm$, and the two $\hshuffle$ rules and the $\reverse(\nil)$
one even with $\gterm$.
Only for the main shuffle rule this is non-trivial to see; here we
have the constraint:
$F(h+t+3) + tF(h+t+3) + F(h+t+3) + [h+hF(h+t+3)+3F(h+t+3)] + 2 > %\\
F(h)+tF(t+1) + F(t+1)
$.
This holds by weak monotonicity of $F$: since $h+t+3 \geq h$ always
holds, we must have $F(h+t+3) \geq F(h)$ as well, and similarly we see
that $tF(h+t+3) \geq tF(t+1)$ and $F(h+t+3) \geq F(t+1)$.
\end{example}

\section{Automation}\label{sec:automation}

To demonstrate that the approach is automatable, we have made a
proof-of-concept implementation of polynomial interpretations in the
higher-order termination tool \wanda.  The implementation only tries
simple parametric shapes, does not use heuristics, and is limited
to second-order AFSs -- a limitation which excludes but 5 out of the
156 higher-order benchmarks in the current \emph{termination problem
database (TPDB)},\footnote{See
  \url{http://termination-portal.org/wiki/TPDB} for
  details on this standard database.}
as the class of second-order systems is very common.\footnote{The
restriction to second-order systems is not essential, but it makes
the code easier in a number of places: we can avoid representing
function-polynomials $\fatlambda \vec{\avar}.P(\vec{\avar})$, stick
to simple interpretation shapes, and we do not have $\max$ in
the left-hand side of constraints.  Mostly, the
restriction is present because of the low number of available
benchmarks of order 3 or higher, which makes it hard to select
suitable interpretation shapes, and not initially worth the added
implementation effort.}
Even with this minimal implementation, the combination of polynomial
interpretations with dependency pairs can handle about 75\% of the
TPDB.

To find polynomial interpretations automatically, \wanda\ uses the
following steps:
\begin{enumerate}
\item assign every function symbol a higher-order
  polynomial with \emph{parameters} as coefficients;
\item for all requirements $l \geqorgterm r$ and $l \geqterm r$,
  calculate $\algintc{l}$ and $\algintc{r}$ as a function on parameters
  and variables -- this gives constraints $P_i \geqorg Q_i$
  and $P_i \geq Q_i$;
\item introduce a parameter $o_i$ for all constraints of the
  form $P_i \geqorg Q_i$, and replace these constraints by
  $P_i \geq Q_i + o_i$; if we also introduce the constraints
  $o_1 + \ldots + o_n \geq 1$ then, when all constraints are
  satisfied, at least one $\geqorg$ constraint is strictly
  oriented;
\item simplify the constraints until they no longer contain
  variables;
\item impose maximum values
  on
  the search space of the parameters
  and use a non-linear constraint solver to find a solution for the
  constraints.
\end{enumerate}

\noindent
These steps are detailed below,
with
an AFS rule for the
function $\map$ as a running
example.

\subsection{Choosing Parametric Polynomial Interpretations}
\label{subsec:wanda:polynomial:choose}

The module for polynomial interpretations
in \wanda\ is called in three
contexts: rule removal, the dynamic dependency pair framework and
the static dependency pair framework.  In the first case, function
interpretations must be strongly monotonic, in the second case they
have to satisfy a subterm property, and in the third there are no
further restrictions.

To start, every function symbol $\afun : [\atype_1 \times \ldots
\times \atype_n] \decpijl \atype_{n+1} \typepijl \ldots \typepijl
\atype_m \typepijl \abasetype \in \F$ is assigned a function of the
form $\fatlambda \avar_1 \ldots \avar_m.p_1 + p_2 + a$, where $a$ is
a parameter and:
\begin{itemize}
\item $p_1$ has the form $a_1 \cdot \avar_1(0,\ldots,0) + \ldots +
  a_m \cdot \avar_m(0,\ldots,0)$, where the $a_i$ are parameters
  (this is well-typed because we work in a second-order
  system);
  \begin{itemize}
  \item in the rule removal setting, we add
    requirements: $a_1 \geq 1, \ldots, a_n \geq 1$;
  \item in the dynamic dependency pairs setting, we add
    requirements: $a_{n+1} \geq 1, \ldots, a_m \geq 1$.
  \end{itemize}
\item $p_2 = q_1 + \ldots + q_k$, where each $q_j$ has the form
  $c_j \cdot \avar_{i_1} \cdots \avar_{i_k} \cdot \avar_j(
  \avar_{i_1}, \ldots, \avar_{i_k}) + d_j \cdot
  \avar_j(\avar_{i_1}, \ldots, \avar_{i_k})$, with $c_j,d_j$ parameters,
  the $\avar_{i_l}$ first-order variables, and $\avar_j$ a
  higher-order variable; every combination of a higher-order variable
  with first-order variables
  occurs.\footnote{\label{foo:nonlin:shape}In case the constraint
    solver does not find a solution for this interpretation shape,
    \wanda\ additionally includes non-linear monomials
    $c_{i,j} \cdot \avar_{i} \cdot \avar_{j}$ (where $i < j$)
    without functional variables in the
    parametric higher-order polynomials and tries again. In general,
    here one can use arbitrary parametric polynomials.}
\end{itemize}

\noindent
We must also choose an interpretation of $@^\atype$ for all
types.  Rather than using a parametric interpretation,
we
observe that
application occurs mostly on the
right-hand side of
constraints.
There, we
often have (sub-)terms $\app{F}{\aterm_1} \cdots \aterm_n$ with $F$ a
free variable; on the left-hand side, such subterms do not occur,
nor can we have applications headed by an abstraction or bound
variable (in a second-order system, bound variables have
base type).  Only applications of the form $\app{\afun(\aterm_1,
\ldots,\aterm_n)}{\aterm_{n+1}} \cdots \aterm_m$ occur on the left;
since function symbols usually have a base type as output type, this
is a rare situation.
Thus, we fix the interpretation of $@^\atype$ for all types to be as
small as possible.  Note that we must have
$\constvaluation(@^\atype) \geqwm \fatlambda f n.f(n)$ by
Theorem~\ref{thm:weakredpair}, and
$\constvaluation(@^\atype)$ may have to be strongly monotonic, or
satisfy a subterm property.
\begin{itemize}
\item in the rule removal setting, $\constvaluation(@^\atype) =
  \fatlambda f n \vec{m}.f(n,\vec{m})+n(\vec{0})$;
\item in the dynamic dependency pairs setting, $\constvaluation(
  @^\atype) = \fatlambda f n \vec{m}.\max(f(n,\vec{m}),n(\vec{0}))$;
\item in the static dependency pairs setting, $\constvaluation(
  @^\atype) = \fatlambda f n \vec{m}.f(n,\vec{m})$.
\end{itemize}
In the rule removal setting this choice together with the constraints
on the parameters guarantees that all $\constvaluation(\afun)$ are
strongly monotonic in the arguments required by the definition of an
extended monotonic algebra and Theorem~\ref{thm:strongredpair}.
In the dynamic dependency pairs setting, we obtain the required
subterm property as demonstrated in Example~\ref{ex:dynamicdp}.
Moreover, in this setting always $\algintc{\app{\afun(
\aterm_1,\ldots,\aterm_n)}{\aterm_{n+1}} \cdots \aterm_m} =
\max(\constvaluation(\afun)(\algintc{\aterm_1},\ldots,
\algintc{\aterm_n},\algintc{\aterm_{n+1}},\ldots,
\linebreak
\algintc{\aterm_m}),
\algintc{\aterm_{n+1}}(\vec{0}),\ldots,\algintc{\aterm_m}(\vec{0})) =
\constvaluation(\afun)(\algintc{\aterm_1},\ldots,\algintc{\aterm_m})$
by the restriction on the parameters.  Thus, although we now also need
to deal with the
$\max$-operator, it will only ever occur on
the right-hand side of a constraint!  Since this avoids the need for
conditional constraints as used in~\cite{fuh:gie:mid:sch:thi:zan:08}
(without losing any power), it both simplifies the automation and
creates smaller constraints.

\smallskip
From these parametric higher-order polynomials, we calculate the
interpretations of terms, and simplify the resulting
higher-order polynomials into a sum of monomials.
For the constraints $l \geqorgterm r$, in general we use constraints
$\algintc{l} \geq \algintc{r} + o$ for
some fresh \emph{bit} $o$ (a parameter whose value ranges over
$\{0,1\}$), and require that
the sum
of these bits is positive.

\begin{example}[Running Example]\label{ex:map}
To
demonstrate the technique,
consider rule removal on the recursive rule of the common $\map$ example,
which gives the constraint
$\map(F,\cons(h,t)) \geqorgterm \linebreak
\cons(\app{F}{h},\map(F,t))$.
We assign:
$
\constvaluation(\cons) =
  \fatlambda n m.a_1 \cdot n + a_2 \cdot m + a_3$ and
$\constvaluation(\map) =
  \fatlambda f n.a_4 \cdot f(0) + a_5 \cdot n +
  a_6 \cdot n \cdot f(n) +
  a_7$.\footnote{To ease presentation, in contrast to \wanda\ here we
    do not use an addend $a_i \cdot f(n)$ for $\map$.}
This leads to the following constraints:
\begin{itemize}
\item $a_1,a_2,a_4,a_5 \geq 1,\, o_1 \geq 1$ (we could also
  immediately replace $o_1$ by 1).
\item $a_7 +
  a_3 \cdot a_5 +
  a_1 \cdot a_5 \cdot h +
  a_2 \cdot a_5 \cdot t +
  a_4 \cdot F(0) +
  a_1 \cdot a_6 \cdot h \cdot F(a_1 \cdot h + a_2 \cdot t + a_3) +
  a_2 \cdot a_6 \cdot t \cdot F(a_1 \cdot h + a_2 \cdot t + a_3) +
  a_3 \cdot a_6 \cdot F(a_1 \cdot h + a_2 \cdot t + a_3)
  \geq
  a_3 +
  a_2 \cdot a_7 +
  o_1 +
  a_1 \cdot h +
  a_2 \cdot a_5 \cdot t +
  a_2 \cdot a_4 \cdot F(0) +
  a_1 \cdot F(h) +
  a_6 \cdot t \cdot F(t)$
\end{itemize}
\end{example}

\subsection{Simplifying Polynomial Requirements}
\label{subsec:wanda:polynomial:simplify}

We obtain requirements that contain variables as well as
parameters; they should be read as ``there exist $a_i,o_k$ such that
for all $h,t,F$ the inequalities hold''.
To avoid dealing with
claims over all
possible numbers or functions
we simplify the requirements until they
contain no more variables.
To a large extent,
these simplifications
correspond to the ones used with automations of polynomial
interpretations for first-order rewriting~\cite{con:mar:tom:urb:05},
but
higher-order variables in
function application
present an extra difficulty.  To deal with
application of higher-order variables, we
will use Lemma~\ref{lem:wanda:polynomial:nasty}:

\begin{lemma}\label{lem:wanda:polynomial:nasty}
Let $F$ be a weakly monotonic functional
and all $p,q,p_i,q_i,s_i,r_i$ polynomials.
\begin{enumerate}
\item\label{it:wanda:polynomial:nasty:product}
  $F(r_1,\ldots,r_k) \cdot p \geq F(s_1,\ldots,s_k) \cdot q$ if
  $r_1 \geq s_1, \ldots, r_k \geq s_k, p \geq q$.
\item\label{it:wanda:polynomial:nasty:complete}
  $r_1 \cdot p_1 + \ldots + r_n \cdot p_n \geq s_1 \cdot q_1 +
  \ldots + s_m \cdot q_m$ if there are $e_{i,j}$ for
  $1 \leq i \leq n,1 \leq j \leq m$ with:
  \begin{enumerate}
  \item for all $i$: $r_i \geq e_{i,1} + \ldots + e_{i,m}$;
  \item for all $j$: $e_{1,j} + \ldots + e_{n,j} \geq s_j$;
  \item either $e_{i,j} = 0$ or $p_i \geq q_j$.
  \end{enumerate}
\end{enumerate}
\end{lemma}

\begin{proof}
\textbf{(\ref{it:wanda:polynomial:nasty:product})} holds by weak
monotonicity of $F$.
As for \textbf{(\ref{it:wanda:polynomial:nasty:complete})},
$r_1 \cdot p_1 + \ldots + r_n \cdot p_n \geq \sum_{i=1}^n
\sum_{j = 1}^m e_{i,j} \cdot p_i$ by \textbf{(a)}, and since
$e_{i,j} = 0$ whenever not $p_i \geq q_j$ by \textbf{(c)},
$\sum_{i=1}^n \sum_{j = 1}^m e_{i,j} \cdot p_i \geq
\sum_{i=1}^n \sum_{j = 1}^m e_{i,j} \cdot q_j =
\sum_{j=1}^m \sum_{i = 1}^n e_{i,j} \cdot q_j$.
Using \textbf{(b)},
$\sum_{j=1}^m \sum_{i = 1}^n e_{i,j} \cdot q_j \geq
\sum_{j = 1}^m s_j \cdot q_j$ as required.
\end{proof}

\noindent
Lemma~\ref{lem:wanda:polynomial:nasty}, together with some
observations used in
the first-order case,
supplies
the theory we need to simplify the requirements to constraints which
do not contain any variables.
Here, a ``component'' of a monomial $a_1 \cdots a_n$ is any of
the $a_i$ (but $p$ is not a component of $F(p)$).

\begin{enumerate}
\item\label{it:wanda:polynomial:simplify}\label{it:wanda:polynomial:trivial}
  Do standard simplifications on the
  constraints, for
  instance replacing
  $3 \cdot F(n) \geq F(n) + a_1 \cdot n$ by $2 \cdot F(n) \geq a_1
  \cdot n$ and
  $p + B \cdot p$ by $(B + 1) \cdot p$ if $B$ is a
  \emph{known} constant, and removing monomials $0 \cdot p$.
  Remove constraints $p \geq 0$ and $p \geq p$ which always hold.
\item\label{it:wanda:polynomial:zero}
  Split constraints $0 \geq p_1 + \ldots + p_n$ into the $n$
  constraints $0 \geq p_i$.
  Remove constraints $0 \geq a$ where $a$ is a single parameter, and
  replace $a$ by $0$ everywhere else. \\
  \emph{This is valid because, in the natural numbers, $0 \geq a$
  implies $a = 0$, and $0+\ldots+0 = 0$.}
\item\label{it:wanda:polynomial:maxsplit}
  Replace constraints $P \geq Q[\max(r,s)]$ by the two constraints
  $P \geq Q[r]$ and $P \geq Q[s]$. \\
  \emph{This is valid because for any valuation $Q[\max(r,s)]$ equals
    $Q[r]$ or $Q[s]$.}
\item\label{it:wanda:polynomial:splitsum}
  Given a constraint $p_1 + \ldots + p_n \geq p_{n+1} + \ldots + p_m$
  where some, but not all, of the monomials $p_i$ contain a
  component $\avar$ or $\avar(\vec{q})$ for some fixed variable
  $\avar$, let $A$ contain the indices $i$ of those monomials $p_i$
  which have $\avar$ or $\avar(\vec{q})$.
  Replace the constraint by the two constraints
  $\sum_{i \in A, i \leq n}\ p_i \geq \sum_{i \in A, i > n}\ p_i$ and
  $\sum_{i \notin A, i \leq n}\ p_i \geq \sum_{i \notin A, i > n}\ 
  p_i$.
  For example, splitting on $n$, the constraint
  $3 \cdot n \cdot m + a_2 \cdot F(a_3 \cdot n + a_4) \geq 2 + a_7
  \cdot m + F(n)$ is split into $3 \cdot n \cdot m \geq 0$ and $a_2
  \cdot F(a_3 \cdot n + a_4) \geq 2 + a_7 \cdot m + F(n)$;
  subsequently, splitting on $F$, the latter is split into $a_2 \cdot
  F(a_3 \cdot n + a_4) \geq F(n)$ and $0 \geq 2 + a_7 \cdot m$. \\
  \emph{This is valid because $p_1 + p_2 \geq q_1 + q_2$ certainly
  holds if $p_1 \geq q_1$ and $p_2 \geq q_2$.}
\item\label{it:wanda:polynomial:divide}
  If all non-zero monomials on either side of a constraint have a
  component $\avar$, ``divide out'' $\avar$.  For example,
  replace the constraint $a_1 \cdot n + n \cdot n \cdot f(a_3,n) \geq
  n + a_3 \cdot n$ by $a_1 + n \cdot f(a_3,n) \geq 1 + a_3$, and
  replace $0 \geq a_5 \cdot m$ by $0 \geq a_5$. \\
  \emph{This is valid because $p \cdot n \geq q \cdot n$
  holds if $p \geq q$ (cf.\ the absolute positiveness
  criterion~\cite{hon:jak:98:1}).}
\item\label{it:wanda:polynomial:product}
  Replace a constraint
  $s \cdot \avar_1(p_{1,1},\ldots,p_{1,k_1}) \cdots \avar_n(p_{n,1},
  \ldots,p_{n,k_n}) \geq s \cdot \avar_1(q_{1,1},\ldots,q_{1,k_1})
  \cdots$\\$\avar_n(q_{n,1},\ldots,q_{n,k_n})$
  by the constraints $s \cdot p_{i,j} \geq s \cdot q_{i,j}$ for all
  $i,j$. \\
  \emph{This is valid by Lemma~\ref{lem:wanda:polynomial:nasty}(\ref{it:wanda:polynomial:nasty:product})
  and case analysis whether $s = 0$ or not.}
\item\label{it:wanda:polynomial:nasty}
  Let $p_1,\ldots,p_n,q_1,\ldots,q_m$ be monomials of the form
  $\avar_1(\vec{r_1}),\ldots,\avar_k(\vec{r_k})$, for fixed
  $\avar_1,\ldots,\avar_k$.
  Replace a constraint $r_1 \cdot p_1 + \ldots + r_n \cdot p_n
    \geq s_1 \cdot q_1 + \ldots + s_m \cdot q_m$ with $n,m \geq 1$ by
    the following constraints, where the $e_{i,j}$ are fresh
    parameters: \\
    \phantom{WW} for $1 \leq i \leq n$: $r_i \geq e_{i,1} + \ldots +
      e_{i,m}$ \\
    \phantom{WW} for $1 \leq j \leq m$: $e_{1,j} + \ldots + e_{n,j}
      \geq s_j$ \\
    \phantom{WW} for $1 \leq i \leq n,\ 1 \leq j \leq m$: $e_{i,j}
      \cdot p_i \geq e_{i,j} \cdot q_j$ (which can
      be handled with clause~\ref{it:wanda:polynomial:product})\\
    \emph{This is valid by
    Lemma~\ref{lem:wanda:polynomial:nasty}(\ref{it:wanda:polynomial:nasty:complete}).}\footnote{In
      the cases where $n = 1$ or $m = 1$, some of these parameters are
      unnecessary; for instance, if $n = 1$, we can safely fix
      $e_{1,j} = s_j$ for all $j$.  Our actual implementation uses
      a few of such special-case optimisations.}
\end{enumerate}

\noindent
It is easy to see that while a constraint still has variables
in it, we can apply clauses to
simplify or split it (taking into account
that $\max$ does not appear in the left-hand side of a
constraint), and that the clauses also terminate on a system without
variables.
These simplifications are not complete: for example,
a
universally valid
constraint $F(n) \cdot n \geq F(1) \cdot n$
is split into constraints $n \geq n$ (which holds), and
$n \geq 1$ (which does not).

\begin{example}
Let us simplify the constraints from Example~\ref{ex:map}.  First,
using clause~\ref{it:wanda:polynomial:splitsum} to group monomials by
their variables, we obtain:
\begin{itemize}
\item $a_1,a_2,a_4,a_5,o_1 \geq 1$
\item $a_7 + a_3 \cdot a_5 \geq a_3 + a_2 \cdot a_7 + o_1$
\item $a_1 \cdot a_5 \cdot h \geq a_1 \cdot h$
\item $a_2 \cdot a_5 \cdot t \geq a_2 \cdot a_5 \cdot t$
\item $a_4 \cdot F(0) +
       a_3 \cdot a_6 \cdot F(a_1 \cdot h + a_2 \cdot t + a_3)
       \geq
       a_2 \cdot a_4 \cdot F(0) +
       a_1 \cdot F(h)$
\item $a_1 \cdot a_6 \cdot h \cdot F(a_1 \cdot h + a_2 \cdot t + a_3)
  \geq 0$
\item
 $a_2 \cdot a_6 \cdot t \cdot F(a_1 \cdot h + a_2 \cdot t + a_3) \geq
  a_6 \cdot t \cdot F(t)$
\end{itemize}

The $4^{\mathrm{th}}$ and $6^{\mathrm{th}}$
requirements are trivial and can be removed with
clause~\ref{it:wanda:polynomial:trivial}.  After dividing
away the non-functional variables using
clause~\ref{it:wanda:polynomial:divide} we have the following
constraints left:
\begin{itemize}
\item $a_1,a_2,a_4,a_5,o_1 \geq 1$
\item $a_7 + a_3 \cdot a_5 \geq a_3 + a_2 \cdot a_7 + o_1$
\item $a_1 \cdot a_5 \geq a_1$
\item $a_4 \cdot F(0) +
       a_3 \cdot a_6 \cdot F(a_1 \cdot h + a_2 \cdot t + a_3)
       \geq
       a_2 \cdot a_4 \cdot F(0) +
       a_1 \cdot F(h)$
\item
 $a_2 \cdot a_6 \cdot F(a_1 \cdot h + a_2 \cdot t + a_3) \geq
  a_6 \cdot F(t)$
\end{itemize}
The first three are completely simplified.
Clauses~\ref{it:wanda:polynomial:nasty} and~\ref{it:wanda:polynomial:product} replace the last two
constraints by:
\begin{itemize}
\item
$a_4 \geq e_{1,1} + e_{1,2}$,\quad
$a_3 \cdot a_6 \geq e_{2,1} + e_{2,2}$,\quad
$e_{1,1} + e_{2,1} \geq a_2 \cdot a_4$,\quad
$e_{1,2} + e_{2,2} \geq a_1$
\item
$e_{1,1} \cdot 0 \geq e_{1,1} \cdot 0$,\quad
$e_{1,2} \cdot 0 \geq e_{1,2} \cdot h$
\item
$e_{2,1} \cdot (a_1 \cdot h + a_2 \cdot t + a_3) \geq e_{2,1} \cdot 0$,\quad
$e_{2,2} \cdot (a_1 \cdot h + a_2 \cdot t + a_3) \geq e_{2,2} \cdot h$
\item
  $a_2 \cdot a_6 \geq k_{1,1}$, \quad $k_{1,1} \geq a_6$, \quad
  $k_{1,1} \cdot (a_1 \cdot h + a_2 \cdot t + a_3) \geq k_{1,1} \cdot t$
\end{itemize}

Using clauses~\ref{it:wanda:polynomial:simplify},
\ref{it:wanda:polynomial:splitsum} and
\ref{it:wanda:polynomial:divide},
we can simplify the constraints
further,
and obtain:
\vspace{-1ex}
\[
\begin{array}{ccc}
\begin{array}{rcl}
a_1,a_2,a_4,a_5,o_1 & \geq & 1 \\
a_7 + a_3 \cdot a_5 & \geq & a_3 + a_2 \cdot a_7 + o_1 \\
a_1 \cdot a_5 & \geq & a_1 \\
a_4 & \geq & e_{1,1} \\
\end{array} &
\begin{array}{rcl}
a_3 \cdot a_6 & \geq & e_{2,1} + e_{2,2} \\
e_{1,1} + e_{2,1} & \geq & a_2 \cdot a_4 \\
e_{2,2} & \geq & a_1 \\
e_{2,2} \cdot a_1 & \geq & e_{2,2} \\
\end{array} &
\begin{array}{rcl}
a_2 \cdot a_{6} & \geq & k_{1,1} \\
k_{1,1} & \geq & a_6 \\
k_{1,1} \cdot a_2 & \geq & k_{1,1} \\
& & \\
\end{array}
\end{array}
\]
\end{example}
\vspace{-1ex}

\noindent
Thus, using a handful of clauses, the requirements are simplified to
a number of constraints with parameters over the natural numbers.  In
the actual \wanda\ implementation a few small optimisations
are used; for example, some simplifications are combined,
and if $\max(r,s)$ occurs more than once in the same polynomial, all
occurrences are replaced by $r$ or $s$ at the same time.
However, these optimisations make no fundamental difference to the
method.

After imposing bounds on the search space,
we can solve the resulting non-linear constraints using standard
SAT- \cite{fuh:gie:mid:sch:thi:zan:07} or SMT-based
\cite{bor:luc:oli:rod:rub:12} techniques
(\wanda\ uses a SAT encoding similar to
\cite{fuh:gie:mid:sch:thi:zan:07} with the solver \textsf{MiniSAT}
\cite{een:sor:04} as back-end).
If the problem is satisfiable, the solver returns
values for all parameters, so it is easy to see which requirements have
been oriented with $>$.
For $\map$, the solver could provide for example the
solution $a_7 = e_{2,1} = 0$,
$a_1 = a_2 = a_3 = a_4 = a_6 = e_{1,1} = e_{2,2} = k_{1,1} = o_1 = 1$,
and $a_5 = 2$. This results in the interpretation
$\constvaluation(\cons) =
\fatlambda nm.n+m+1$ and $\constvaluation(\map) = \fatlambda
fn.f(0) + 2 \cdot n + n\cdot f(n)$.

\section{Experiments}\label{sec:experiments}
For an empirical evaluation of our contributions, we conducted a
number of experiments with our implementation in \wanda\ using an Intel
Xeon 5140 CPU with four cores
at 2.33~GHz (cf.\ also
  \url{http://aprove.informatik.rwth-aachen.de/eval/HOPOLO/} for
  details on the evaluation).
As underlying benchmark set, we used the 156 examples from the
higher-order category of the TPDB version 8.0.1 together with Examples
\ref{ex:hshuffledef} and \ref{ex:dynamicdp}. \wanda\ invokes the
SAT solver \minisat\ \cite{een:sor:04} and the first-order termination
\pagebreak
prover \aprove\ \cite{gie:sch:thi:06} as
back-ends.
As in the Termination Competition, we imposed a 60 second timeout
per example.

The module for polynomial interpretation is called potentially twice
with different polynomial shapes, as described at the start of
Section~\ref{subsec:wanda:polynomial:choose}
(cf.\ Footnote~\ref{foo:nonlin:shape}).  The search space
for the parameters is $\{0,\ldots,3\}$.
Our first experiment is designed to analyse the impact of polynomial
interpretations coupled with a higher-order dependency pair
framework.

\begin{figure}[h]\small
\vspace{-8pt}
\begin{center}
\begin{tabular}{r||c|c|c|c|c}
Configuration     & YES & NO & MAYBE & TIMEOUT & Avg.\ time\\\hline
\wanda\ full      & 124 &  9 &   23  &    2    & 3.19 $s$ \\
\wanda\ no poly   & 119 &  9 &   30  &    0    & 2.40 $s$ \\
\wanda\ no horpo  & 118 &  9 &   28  &    3     & 3.59 $s$ \\
\end{tabular}
\end{center}
\vspace{-17pt}
\caption{\small \emph{Experimental results of full \wanda\ with and
without polynomials or horpo}}
\vspace{-10pt}
\label{fig:expfull}
\end{figure}

\noindent
Fig.~\ref{fig:expfull} shows the results of \emph{\wanda\ full},
which includes both polynomial interpretations and HORPO,
the other main class of orderings implemented by \wanda\ (other than
that, \wanda\ only uses the subterm criterion as an ordering-based
technique).  They are compared to versions of \wanda\ where either
polynomials or HORPO are disabled.
Although \wanda\ already scored highest in the Termination Competition
of 2011,
adding the contributions of this paper
gives an additional 5
examples on the benchmark set.
It is interesting to note that even without HORPO, \wanda\ with
polynomials can
still show termination of 118 examples.

Using the contributions of \cite{fuh:kop:11:1}, \wanda\ delegates the
first-order part of a higher-order rewrite system to the
first-order termination tool \aprove, where it is
commonplace to use polynomial interpretations. The setup of our
second experiment deals with the impact of higher-order polynomial
interpretations if \wanda\ does not use a first-order tool.
\begin{figure}[h]\small
\vspace{-8pt}
\begin{center}
\begin{tabular}{r||c|c|c|c|c}
Configuration                           & YES & NO & MAYBE & TIMEOUT & Avg.\ time\\\hline
\wanda\ no~\cite{fuh:kop:11:1} full     & 118 &  9 & 29 &  2 & 2.32 $s$ \\ % proverid 8786
\wanda\ no~\cite{fuh:kop:11:1} no poly  & 107 &  9 & 42 &  0 & 1.09 $s$ \\ % proverid 8779
\wanda\ no~\cite{fuh:kop:11:1} no horpo & 111 &  9 & 35 &  3 & 2.89 $s$ \\ % proverid 8777
\end{tabular}
\end{center}
\vspace{-15pt}
\caption{\small \emph{Experimental results of \wanda\ without
  first-order back-end}}
\vspace{-10pt}
\label{fig:expnofo}
\end{figure}

\noindent
Fig.~\ref{fig:expnofo} juxtaposes the results of \wanda\ without the
first-order prover \aprove\ in three configurations.
We see that
if we disable the 
first-order back-end, the increase in power by polynomial
interpretations goes up from 5 examples in the first experiment to 11
examples in the second.
Thus,
the gain of using a first-order tool can
at least partially be compensated by using native higher-order
polynomial interpretations.

Our third experiment investigates the impact of higher-order
polynomial interpretations if no dependency pairs are used (which also
excludes first-order termination tools).
Here we compare to the version of HORPO implemented in \wanda.
\begin{figure}[h]\small
\vspace{-8pt}
\begin{center}
\begin{tabular}{r||c|c|c|c|c}
Configuration      & YES & NO & MAYBE & TIMEOUT & Avg.\ time\\\hline
Rule Removal both  & 76  &  9 &   70  &  3 & 3.60 $s$ \\ % proverid 8775
Rule Removal horpo & 69  &  9 &   80  &  0 & 1.01 $s$ \\ % proverid 8773
Rule Removal poly  & 47  &  9 &   97  &  5 & 3.64 $s$ \\ % proverid 8772
\end{tabular}
\end{center}
\vspace{-15pt}
\caption{\small \emph{Experimental results of \wanda\ with rule removal
(and without dependency pairs)}}
\vspace{-10pt}
\label{fig:expnodp}
\end{figure}

\noindent
Using just rule removal, HORPO clearly trumps polynomial
interpretations.  However, in part this may be due to the limited
choice in interpretation shapes this first implementation of
polynomial interpretations supports.

\paragraaf{Discussion}
Analysing the termination problem database, it is perhaps not
surprising that the gain from using polynomial interpretations in the
first experiment is not larger: the majority of the benchmarks which
\wanda\ cannot already handle is non-terminating, or not known to be
terminating (for example, state-of-the-art first-order tools cannot
prove termination of the first-order part).  For others,
type-conscious methods such as \emph{accessibility} (see
e.g.~\cite{bla:jou:rub:08:1}) are required; the method described in
this paper ignores differences in base types.
For cases where polynomial interpretations are needed, but only for
the (truly) first-order part, passing this first-order part
\cite{fuh:kop:11:1} to a modern first-order tool already suffices --
as is evident by comparing the numbers in the first and second
experiments.
With higher-order polynomial interpretations, we have gained three out
of the remaining seven benchmarks.

\section{Conclusion}\label{sec:conclusion}

In this paper, we have extended the termination method of weakly
monotonic algebras to the class of AFSs,
simplifying definitions and
adding the theory to use algebras with rule removal and dependency
pairs; some efforts towards this were previously made
in~\cite{kop:raa:11:1}, but only for the setting of dynamic
dependency pairs.  Then, we introduced the class of higher-order
polynomial interpretations, and discussed how suitable interpretations
can be found automatically.  The implementation of polynomial
interpretations increases the power of \wanda\ by a respectable five
benchmarks, including the two examples in this paper.

Thus, weakly monotonic algebras form an elegant method for proving
termination by hand and, as demonstrated by the implementation in
\wanda\ and the results of the experiments, a feasible automatable
termination method as well.

\paragraaf{Future Work}
We have by no means reached the limit of what can be achieved with
this technique: we might consider different interpretation shapes,
possibly coupled with heuristics to determine a suitable shape.  Or
we may go beyond polynomials; we could for instance use $\max$ in
function interpretations as done in
e.g.~\cite{fuh:gie:mid:sch:thi:zan:08}, or (for a truly higher-order
alternative), use repeated function application; this leads to
interpretations like $\fatlambda nmf.\max(m,f^n(m))$.\footnote{The
$\max$ is essential in this interpretation because $\fatlambda nmf.
f^n(m)$ is not weakly monotonic.  A case analysis whether $m \geq
f(m)$ or $f(m) \geq m$ shows that $\fatlambda nmf.\max(m,f^n(m))$
\emph{is} weakly monotonic.}

Another alley to explore
is to combine polynomial
interpretations with type interpretations: rather than collapsing all
base types into one, we might \emph{translate} them,
e.g.\ 
mapping a base type $\mathsf{funclist}$ to $\o \typepijl \o$.
\wanda\ already does this in very specific cases,
and one could
simultaneously
search for polynomial interpretations
and for a
type interpretation -- this
could
parallel
the
search for type orderings
in
implementations of
the recursive path
ordering~\cite{bla:jou:rub:08:1}.

Moreover, in the first-order world, there are many more applications
of monotonic algebras, e.g.\ matrix, arctic, rational, real and integer
interpretations\ldots
There is no obvious reason why these methods cannot be lifted to
the higher-order case as well!

\bibliographystyle{plain}
\bibliography{references}

\begin{thebibliography}{10}

\bibitem{art:gie:00:1}
T.~Arts and J.~Giesl.
\newblock Termination of term rewriting using dependency pairs.
\newblock {\em Theoretical Computer Science}, 236(1-2):133--178, 2000.

\bibitem{bla:jou:rub:08:1}
F.~Blanqui, J.-P. Jouannaud, and A.~Rubio.
\newblock The computability path ordering: The end of a quest.
\newblock In {\em Proc.\ CSL 2008}, LNCS 5213, pages 1--14, 2008.

\bibitem{bor:luc:oli:rod:rub:12}
C.~Borralleras, S.~Lucas, A.~Oliveras, E.~{Rodr\'{\i}guez-Carbonell}, and
  A.~Rubio.
\newblock {SAT} modulo linear arithmetic for solving polynomial constraints.
\newblock {\em Journal of Automated Reasoning}, 48(1):107--131, 2012.

\bibitem{thor}
C.~Borralleras and A.~Rubio.
\newblock \textsf{{THOR}} -- a higher-order termination tool.
\newblock \url{http://www.lsi.upc.edu/~albert/term.html}.

\bibitem{bor:rub:01:1}
C.~Borralleras and A.~Rubio.
\newblock A monotonic higher-order semantic path ordering.
\newblock In {\em Proc.\ LPAR 2001}, LNAI 2250, pages 531--547, 2001.

\bibitem{con:mar:tom:urb:05}
E.~Contejean, C.~March{\'e}, A.~P. Tom{\'a}s, and X.~Urbain.
\newblock Mechanically proving termination using polynomial interpretations.
\newblock {\em Journal of Automated Reasoning}, 34(4):325--363, 2005.

\bibitem{een:sor:04}
N.~E{\'e}n and N.~S{\"o}rensson.
\newblock An extensible {SAT}-solver.
\newblock In {\em Proc.\ SAT 2003}, LNCS 2919, pages 502--518, 2004.

\bibitem{end:12:1}
J.~Endrullis.
\newblock \textsf{{Jambox}}.
\newblock \url{{http://joerg.endrullis.de/}}.

\bibitem{end:wal:zan:08}
J.~Endrullis, J.~Waldmann, and H.~Zantema.
\newblock Matrix interpretations for proving termination of term rewriting.
\newblock {\em Journal of Automated Reasoning}, 40(2-3):195--220, 2008.

\bibitem{fuh:gie:mid:sch:thi:zan:07}
C.~Fuhs, J.~Giesl, A.~Middeldorp, P.~Schneider-Kamp, R.~Thiemann, and H.~Zankl.
\newblock {SAT} solving for termination analysis with polynomial
  interpretations.
\newblock In {\em Proc.\ SAT 2007}, LNCS 4501, pages 340--354, 2007.

\bibitem{fuh:gie:mid:sch:thi:zan:08}
C.~Fuhs, J.~Giesl, A.~Middeldorp, P.~{Schneider-Kamp}, R.~Thiemann, and
  H.~Zankl.
\newblock Maximal termination.
\newblock In {\em Proc.\ RTA 2008}, LCNS 5117, pages 110--125, 2008.

\bibitem{fuh:kop:11:1}
C.~Fuhs and C.~Kop.
\newblock Harnessing first order termination provers using higher order
  dependency pairs.
\newblock In {\em Proc.\ FroCoS 2011}, LNAI 6989, pages 147--162, 2011.

\bibitem{fuh:kop:12:1}
C.~Fuhs and C.~Kop.
\newblock Polynomial interpretations for higher-order rewriting.
\newblock In {\em Proc.\ RTA 2012}, LIPIcs, 2012.
\newblock To appear.

\bibitem{gie:sch:thi:06}
J.~Giesl, P.~Schneider-Kamp, and R.~Thiemann.
\newblock \textsf{AProVE 1.2}: Automatic termination proofs in the dependency
  pair framework.
\newblock In {\em Proc.\ IJCAR 2006}, LNAI 4130, pages 281--286, 2006.

\bibitem{hon:jak:98:1}
H.~Hong and D.~Jaku{\v{s}}.
\newblock Testing positiveness of polynomials.
\newblock {\em Journal of Automated Reasoning}, 21(1):23--38, 1998.

\bibitem{jou:oka:91:1}
J.-P. Jouannaud and M.~Okada.
\newblock A computation model for executable higher-order algebraic
  specification languages.
\newblock In {\em Proc.\ LICS 1991}, pages 350--361, 1991.

\bibitem{jou:rub:07:1}
J.-P. Jouannaud and A.~Rubio.
\newblock Polymorphic higher-order recursive path orderings.
\newblock {\em Journal of the ACM}, 54(1):1--48, 2007.

\bibitem{wanda}
C.~Kop.
\newblock \textsf{{WANDA}} -- a higher order termination tool.
\newblock \url{http://few.vu.nl/~kop/code.html}.

\bibitem{kop:11:1}
C.~Kop.
\newblock Simplifying algebraic functional systems.
\newblock In {\em Proc.\ CAI 2011}, LNCS 6742, pages 201--215, 2011.

\bibitem{kop:raa:11:1}
C.~Kop and F.~van Raamsdonk.
\newblock Higher order dependency pairs for algebraic functional systems.
\newblock In {\em Proc.\ RTA 2011}, LIPIcs 10, pages 203--218, 2011.

\bibitem{kop:raa:12:1}
C.~Kop and F.~van Raamsdonk.
\newblock Dynamic dependency pairs for algebraic functional systems.
\newblock {\em Logical Methods in Computer Science}, 2012.
\newblock Special Issue of the 22nd International Conference on Rewriting
  Techniques and Applications (RTA 2011). To appear.

\bibitem{kor:ste:zan:mid:09}
M.~Korp, C.~Sternagel, H.~Zankl, and A.~Middeldorp.
\newblock \textsf{Tyrolean Termination Tool 2}.
\newblock In {\em Proc.\ RTA 2009}, LNCS 5595, pages 295--304, 2009.

\bibitem{kus:iso:sak:bla:09:1}
K.~Kusakari, Y.~Isogai, M.~Sakai, and F.~Blanqui.
\newblock Static dependency pair method based on strong computability for
  higher-order rewrite systems.
\newblock {\em IEICE Transactions on Information and Systems},
  92(10):2007--2015, 2009.

\bibitem{lan:79}
D.~Lankford.
\newblock On proving term rewriting systems are {N}oetherian.
\newblock Technical Report MTP-3, Louisiana Technical University, Ruston, LA,
  USA, 1979.

\bibitem{nip:91:1}
T.~Nipkow.
\newblock Higher-order critical pairs.
\newblock In {\em Proc.\ LICS 1991}, pages 342--349, 1991.

\bibitem{pol:94:1}
J.~van~de Pol.
\newblock Termination proofs for higher-order rewrite systems.
\newblock In {\em Proc.\ HOA 1993}, LNCS 816, pages 305--325, 1994.

\bibitem{pol:96:1}
J.C. van~de Pol.
\newblock {\em Termination of Higher-order Rewrite Systems}.
\newblock PhD thesis, University of Utrecht, 1996.

\bibitem{sak:wat:sak:01}
M.~Sakai, Y.~Watanabe, and T.~Sakabe.
\newblock An extension of the dependency pair method for proving termination of
  higher-order rewrite systems.
\newblock {\em IEICE Transactions on Information and Systems},
  E84-D(8):1025--1032, 2001.

\bibitem{ter:03}
Terese.
\newblock {\em Term Rewriting Systems}, volume~55 of {\em Cambridge Tracts in
  Theoretical Computer Science}.
\newblock Cambridge University Press, 2003.

\bibitem{termcomp}
Wiki.
\newblock Termination portal.
\newblock \url{http://www.termination-portal.org/}.

\end{thebibliography}

% The appendix only goes to the report version.
\confreport{}{
\appendix

\renewenvironment{lemma}[1]{
  \vspace{1ex}
  \noindent$\blacktriangleright$ \textbf{#1.}\it}{
  \vspace{1ex}
}

\newpage

\section{Appendix}

Something that is worth noting, in particular when considering the
proofs in the following section, is that we use the mathematical
definition of a function as a set of pairs; a function is specified
entirely by its domain and values.
Thus, if $F$ and $G$ are both functions in some
$\WM_\atype$, and $F(x) = G(x)$ for all $x$ in their domain, then
$F = G$.
We will use the notation $\fatlambda \avar.P(\avar)$ for a function
that takes one argument $\avar$, and returns $P(\avar)$.

\subsection{Changing the Definition of $\geq$.}

In Section~\ref{subsec:wmf} we used slightly different restrictions
on the orderings $>$ and $\geq$ than in~\cite{pol:96:1}: van de Pol
required that $\geq$ was the reflexive closure of $>$, while we merely
require that $\geq$ is compatible with $>$.  Is it certain we can do
this?

The answer is yes, it is.  In fact, of the theory in~\cite{pol:96:1}
we use but three results: Lemma~3.2.1 (the substitution lemma),
Lemma~4.1.4 (which gives facts about the interaction of $\gwm$ and
$\geqwm$) and Proposition~4.1.5 (which states that $\algintc{\aterm}$
is a weakly monotonic functional if $\constvaluation$ and
$\varvaluation$ map to weakly monotonic functions).

Lemma~3.2.1 is completely independent of the definition of $\WM$ (in
fact, $\WM$ is only defined a chapter later).
Lemma~4.1.4 is used only for Lemma~\ref{lem:gwminteract}, which we
will rederive below (in fact, Lemma~\ref{lem:gwminteract} as it is
already is not literally what appears in~\cite{pol:96:1}).
As for Proposition~4.1.5, it uses only reflexivity of $\geqwm$, which
remains valid, and the following facts:
\begin{itemize}
\item if $f \in \WM_{\atype \ftypepijl \btype}$ and $x \in
  \WM_\atype$, then $f(x) \in \WM_\btype$;
\item if $f \geqwm_{\atype \ftypepijl \btype} g$ and $x \in
  \WM_\atype$, then $f(x) \geqwm g(x)$;
\item if $f \gwm_{\atype \ftypepijl \btype} g$ and $x \in
  \WM_\atype$, then $f(x) \gwm g(x)$;
\item if $f \in \WM_{\atype \ftypepijl \btype}$ and $x,y \in
  \WM_\atype$ and $x \geqwm y$ then $f(x) \geqwm f(y)$.
\end{itemize}
These facts are all immediately clear from the definition of $\WM$,
and they do not depend on the way $>$ and $\geq$ interact.

\medskip \noindent
The one thing we do have to see is that Lemma~\ref{lem:gwminteract}
stays valid.  Recall:
\begin{itemize}
\item $\geq$ is a quasi-ordering, so a reflexive and transitive
  binary relation;
\item $>$ is a well-founded partial ordering, so a transitive binary
  relation, such that there is no infinite decreasing sequence $a_1 >
  a_2 > \ldots$ ($>$ must also be non-reflexive and anti-symmetric,
  but this is implied by well-foundedness);
\item $>$ and $\geq$ are compatible, so either $a > b \geq c$ implies
  $a > c$, or $a \geq b > c$ implies $a > c$;
\item $\basealgebraset$ is non-empty.
\end{itemize}

\begin{lemma}{Lemma~\ref{lem:gwminteract}}
For all types $\atype$ the following statements hold:
\begin{itemize}
\item $\gwm_\atype$ is well founded;
\item both $\gwm_\atype$ and $\geqwm_\atype$ are transitive;
\item $\geqwm_\atype$ is reflexive (always $n \geqwm_\atype n$);
\item $\gwm_\atype$ and $\geqwm_\atype$ are compatible.
\end{itemize}
\end{lemma}

\begin{proof}
We prove the lemma with induction on the type $\atype$, and in
addition that weakly monotonic functionals exist for all types (a
fact we will need for the other statements).
Assume (IH) that for all subtypes $\btype$ of $\atype$,\ 
$\gwm_\btype$ is well founded, both $\gwm_\btype$ and $\geqwm_\btype$
are transitive, $\geqwm_\btype$ is reflexive and $\gwm_\btype$ and
$\geqwm_\btype$ are compatible \emph{in the same way as $>$ and
$\geq$}, and that $\WM_\btype$ is non-empty.  ``In the same way''
means that if $> \cdot \geq\ \subseteq\ >$ then also $\gwm_\btype
\cdot \geqwm_\btype\ \subseteq\ \gwm_\btype$, and otherwise
$\geqwm_\btype \cdot \gwm_\btype\ \subseteq\ \gwm_\btype$.

For $\atype$ a base type, we immediately have well-foundedness,
transitivity, reflexivity, compatibility and non-emptiness, by the
assumptions on $>,\ \geq$ and $\basealgebraset$.
For $\atype = \btype \typepijl \ctype$, we obtain:

\textbf{$\gwm_\atype$ is well founded:}
Suppose, towards a contradiction, that $f_1 \gwm_\atype f_2
\gwm_\atype f_3 \gwm_\atype \ldots$.  Let $a \in \WM_\btype$
(such $a$ exists by IH).  By definition of $\gwm_\atype$, also
$f_1(a) \gwm_\ctype f_2(a) \gwm_\ctype f_3(a) \gwm_\ctype \ldots$,
contradicting well-foundedness of $\gwm_\ctype$.

\textbf{$\geqwm_\atype$ is transitive:}
Suppose $f \geqwm_\atype g \geqwm_\atype h$.  Then for all
$x \in \WM_\btype$ we have $f(x) \geqwm_\ctype g(x) \geqwm_\ctype
h(x)$ by definition of $\geqwm_\atype$, so by the induction
hypothesis $f(x) \geqwm_\ctype h(x)$ for all $x \in \WM_\btype$.
This exactly means that $f \geqwm_\atype h$.

\textbf{$\gwm_\atype$ is transitive:}
Same as for $\geqwm_\atype$.

\textbf{$\geqwm_\atype$ is reflexive:}
Let $f \in \WM_\atype$.  Then $f \geqwm_\atype f$ iff for all $x \in
\WM_\btype$: $f(x) \geqwm_\ctype f(x)$.  But this holds by
reflexivity of $\geqwm_\ctype$ (IH).

\textbf{$\gwm_\atype$ and $\geqwm_\atype$ are compatible:}
Suppose $> \cdot \geq$ is included in $>$; the case when $\geq \cdot
>$ is included in $>$ is symmetric.
Let $f,g,h \in \WM_\atype$ and suppose $f \gwm_\atype g \geqwm_\atype
h$.  Then for all $x \in \WM_\atype$ we have $f(x) \gwm_\ctype g(x)
\geqwm_\ctype h(x)$, so $f(x) \gwm_\ctype h(x)$ by (IH), and therefore
$f \gwm_\atype h$.

\textbf{$\WM_\atype$ is non-empty:}
Let $a \in \WM_\ctype$; the function $g := \fatlambda n : \WM_\btype.
a$ is in $\WM_\atype$ because if $x \geqwm_\btype y$, then
$g(x) = a \geqwm a = g(y)$ by reflexivity of $\geqwm_\ctype$ (IH).
\end{proof}

\subsection{The Max function.}
The $\max_\atype$ function is defined for all $\atype$ in
Example~\ref{exa:wmf}.  The other two parts of this example (the
constant function and maximum function) are presented as weakly
monotonic functionals already in~\cite{pol:96:1}, but the $\max$
function is new, so it falls on us to demonstrate its weak
monotonicity.

\begin{lemma}{Example~\ref{exa:wmf}(\ref{ex:wmf:max})}
$\max_\atype \in \WM_{\atype \ftypepijl \abasetype \ftypepijl \atype}$
for all types $\atype$ and base types $\abasetype$.
\end{lemma}

\begin{proof}
Write $\atype = \atype_1 \typepijl \ldots \typepijl
\atype_k \typepijl \abasetype$ with $\abasetype \in \setsorts$.
Using the observation above Lemma~\ref{lem:gwminteract}, it suffices
if $\max_\atype(f,n,m_1,\ldots,m_k) \geq \max_\atype(f',n',m_1',
\ldots,m_k')$ if $f \geqwm f',\ n \geqwm n'$ and each $m_i \geqwm
m_i'$ (where $f,f' \in \WM_\atype, n,n' \in \basealgebraset$ and each
$m_i,m_i' \in \WM_{\atype_i}$).

But
$\max_\atype(f,n,m_1,\ldots,m_k) = \max(f(m_1,\ldots,m_n),n)$ and
$\max_\atype(f',n',m_1',\ldots,m_k') = \max(f'(m_1',\ldots,m_k'),n')$;
certainly $n \geq n'$; by definition of $\max$ we are done if
$f(m_1,\ldots,m_k) \geqwm f'(m_1',\ldots,m_k')$.

By induction on $i$ we have: $f(m_1,\ldots,m_i) \in \WM_{\atype_{i+1}
\ftypepijl \ldots \ftypepijl \atype_k \ftypepijl \abasetype}$ for all
$0 \leq i \leq k$ and $f(m_1',\ldots,m_i') \geqwm f'(m_1',\ldots,
m_i')$:
\begin{itemize}
\item in the base case ($i = 0$), $f(m_1,\ldots,m_i) = f \geqwm f' =
  f'(m_1',\ldots,m_i')$ by assumption
\item if $i = j+1$, then $f(m_1,\ldots,m_i) = f(m_1,\ldots,m_j)(m_i)$,
  and since $f(m_1,\ldots,m_j) \in \WM_{\atype_i \ftypepijl
  \ldots \ftypepijl \atype_k \ftypepijl \abasetype}$ by the induction
  hypothesis, and $m_i \in \WM_i$ by assumption, this functional is
  in $\WM_{\atype_{i+1} \ftypepijl \ldots \ftypepijl \atype_k
  \ftypepijl \abasetype}$ by definition
\item
  $f(m_1,\ldots,m_j)(m_i) \geqwm
  f(m_1,\ldots,m_j)(m_i')$ by the definition of $\geqwm$, and since
  $f(m_1,\ldots,m_j) \geqwm f'(m_1',\ldots,m_j')$ by the induction
  hypothesis, we obtain: \\
  $f(m_1,\ldots,m_i) \geqwm f(m_1,\ldots,m_j)(m_i') \geqwm
  f'(m_1',\ldots,m_j')(m_i') = f'(m_1',\ldots,m_i')$
\end{itemize}
Taking $i := k$ this provides what we need.
\end{proof}

\subsection{Weak and Strong Monotonicity: claims in
Lemma~\ref{lem:wmf:afs}(\ref{lem:wmf:afsm:context}) and
Theorem~\ref{thm:strongredpair}}

\begin{lemma}{Lemma~\ref{lem:wmf:afs}(\ref{lem:wmf:afsm:context})}
Let $(\basealgebra, \constvaluation)$ be a weakly monotonic algebra.
If $\algintc{\aterm} \geqwm \algintc{\bterm}$ for all valuations
$\varvaluation$, then $\algintc{C[\aterm]} \geqwm \algintc{C[\bterm]}$
for all valuations $\varvaluation$ and contexts $C$.
\end{lemma}

\noindent
This Lemma has no counterpart in~\cite{pol:96:1}, because van de Pol
does not consider situations where weak monotonicity is sufficient.
Thus, this we derive ourselves.  The proof is an easy induction.

\begin{proof}
By induction on the form of $C$.
The base case ($C = \Box_\atype$) is evident, otherwise suppose (IH)
$\algint{D[\aterm]}_{\constvaluation,\varvaluationb} \geqwm
\algint{D[\bterm]}_{\constvaluation,\varvaluationb}$ for all
valuations $\varvaluationb$.

In the case of an abstraction, $C[] = \abs{\avar}{D[]}$, we have
$\algintc{C[\aterm]} = \fatlambda n.\algint{D[\aterm]}_{
\constvaluation,\varvaluation \cup \{\avar \mapsto n\}}$ and we are
done because, by (IH) and the definition of $\geqwm$ for functions,
this function $\geqwm \fatlambda n.\algint{D[\bterm]}_{
\constvaluation,\varvaluation \cup \{\avar \mapsto n\}} =
\algintc{C[\bterm]}$.

In the case of a function application, $C[] = \afun(\aterm_1,\ldots,
D[],\ldots,\aterm_n)$, we have $\algintc{C[\aterm]} =
\constvaluation(\afun)(\algintc{\aterm_1},\ldots,\algintc{D[\aterm]},
\ldots,\algintc{\aterm_n})$.  Weak monotonicity of $\constvaluation(
\afun)$ implies that if any of the argument $\geqwm$-decreases, then
so does the result.  Thus, by (IH) also $\algintc{C[\aterm]} \geqwm
\constvaluation(\afun)(\algintc{\aterm_1},\ldots,\algintc{D[\bterm]},
\ldots,\algintc{\bterm_1}) = \algintc{C[\bterm]}$.

The cases where $C[] = \app{D[]}{\cterm}$ or $C[] = \app{\cterm}{
D[]}$ are very similar.
\end{proof}

Almost the same, but using strong monotonicity, the following was
stated in Theorem~\ref{thm:strongredpair} as an easily derived result.

\begin{lemma}{Theorem~\ref{thm:strongredpair}(claim)}
Let $(\basealgebra, \constvaluation)$ be an extended monotonic
algebra.
If $\algintc{\aterm} \gwm \algintc{\bterm}$ for all valuations
$\varvaluation$, then $\algintc{C[\aterm]} \gwm \algintc{C[\bterm]}$
for all valuations $\varvaluation$ and contexts $C$.
\end{lemma}

\begin{proof}
By induction on the form of $C$.
The base case ($C = \Box_\atype$) is evident, otherwise suppose (IH)
$\algint{D[\aterm]}_{\constvaluation,\varvaluationb} \gwm
\algint{D[\bterm]}_{\constvaluation,\varvaluationb}$ for all
valuations $\varvaluationb$.

In the case of an abstraction, $C[] = \abs{\avar}{D[]}$, we have
$\algintc{C[\aterm]} = \fatlambda n.\algint{D[\aterm]}_{
\constvaluation,\varvaluation \cup \{\avar \mapsto n\}}$ and we are
done because, by (IH) and the definition of $\gwm$ for functions,
this function $\gwm \fatlambda n.\algint{D[\bterm]}_{
\constvaluation,\varvaluation \cup \{\avar \mapsto n\}} =
\algintc{C[\bterm]}$.

In the case of a function application, $C[] = \afun(\aterm_1,\ldots,
D[],\ldots,\aterm_n)$, let $\constvaluation(\afun) = A$; a function
which is strongly monotonic in its first $n$ arguments by assumption.
We have $\algintc{C[\aterm]} =
A(\algintc{\aterm_1},\ldots,\algintc{D[\aterm]},\ldots,
\algintc{\aterm_n})$.  By the induction hypothesis,
$\algintc{D[\aterm]} \gwm \algintc{D[\bterm]}$, so because $A$ is
strongly monotonic in the corresponding argument 
$\algintc{C[\aterm]} \gwm A(\algintc{\aterm_1},\ldots,
\algintc{D[\bterm]},\ldots, \algintc{\bterm_1}) =
\algintc{C[\bterm]}$.

The cases where $C[] = \app{D[]}{\cterm}$ or $C[] = \app{\cterm}{
D[]}$ are very similar, since $@^\atype$ is also strongly monotonic
in its first two arguments.
\end{proof}

\subsection{Example~\ref{ex:dynamicdp}: dealing with dynamic
dependency pairs}

In Example~\ref{ex:dynamicdp}, we made the claim that, choosing
$\constvaluation(
@^{\atype \ftypepijl \btype}) = \fatlambda fn.\max_\btype(f(n),n(
\vec{0}))$ and $\constvaluation(c_j) = 0_{\atype}$ for $c_j :
\atype$, we have $\algintc{\app{\aterm}{\vec{\bterm}}} \geqwm
\algintc{\app{\bterm_i}{\vec{c}}}$.

To see that this holds, consider the following lemma:

\begin{lemma}{Lemma: Interpretation of Application}
for all $n$ we have:
\[\algintc{\app{\aterm}{\bterm_1} \cdots \bterm_n} =
\fatlambda \vec{m}.\max(\algintc{\aterm}(\algintc{\bterm_1},\ldots,
\algintc{\bterm_n},\vec{m}), \algintc{\bterm_1}(\vec{0}),\ldots,
\algintc{\bterm_n}(\vec{0}))\]
\vspace{-6ex}
\end{lemma}

\begin{proof}
We see this with induction on $n$.

\emph{$n = 0$ (base case):} $\algintc{\aterm} = \fatlambda \vec{m}.
\algintc{\aterm}(\vec{m})$, since functions are defined by their
values.

\emph{$n = k+1$:} $\algintc{\app{\aterm}{\bterm_1} \cdots \bterm_n} =
\max_\btype(\algintc{\app{\aterm}{\bterm_1} \cdots \bterm_k}(
\algintc{\bterm_n}),\algintc{\bterm_n}(\vec{0}))$, which by the
definition of $\max_\btype$ equals
$\fatlambda \vec{m}.\max(\algintc{\app{\aterm}{\bterm_1} \cdots
\bterm_k}(\algintc{\bterm_n},\vec{m}),\algintc{\bterm_n}(\vec{0}))$.
By the induction hypothesis, \\
\indent \indent $\algintc{\app{\aterm}{\bterm_1} \cdots
  \bterm_k}(\algintc{\bterm_n},\vec{m})$ \\
\indent \indent $= [\fatlambda \vec{x}.\max(\algintc{\aterm}(\algintc{
\bterm_1},\ldots,\algintc{\bterm_k},\vec{x}),\algintc{\bterm_1}(
\vec{0}),\ldots,\algintc{\bterm_k}(\vec{0}))](\algintc{\bterm_n},
\vec{m})$ \\
\indent \indent $= \max(\algintc{\aterm}(\algintc{\bterm_1},\ldots,
\algintc{\bterm_n},\vec{m}),\algintc{\bterm_1}(\vec{0}),\ldots,
\algintc{\bterm_k}(\vec{0}))$. \\
Thus, the function we have is exactly: \\
\indent \indent
$\fatlambda \vec{m}.\max(\max(\algintc{\aterm}(\algintc{\bterm_1},
\ldots,\algintc{\bterm_n},\vec{m}),
\algintc{\bterm_1}(\vec{0}),\ldots,\algintc{\bterm_k}(\vec{0})),
\algintc{\bterm_n}(\vec{0}))$ \\
\indent \indent
$= \fatlambda \vec{m}.\max(\algintc{\aterm}(\algintc{\bterm_1},
\ldots,\algintc{\bterm_n},\vec{m}),\algintc{\bterm_1}(\vec{0}),
\ldots,\algintc{\bterm_n}(\vec{0}))$
\end{proof}

Thus we see, if $\app{\aterm}{\vec{\bterm}}$ and $\app{\bterm_i}{
\vec{c}}$ both have base type, then:
\[
\begin{array}{lll}
\algintc{\app{\aterm}{\vec{\bterm}}} & = &
\max(\algintc{\aterm}(\algintc{\bterm_1},\ldots,
  \algintc{\bterm_n}),\algintc{\bterm_1}(\vec{0}),\ldots,
  \algintc{\bterm_n}(\vec{0})) \\
& \geq & \algintc{\bterm_i}(\vec{0}) \\
& = & \max(\algintc{\bterm_i}(\vec{0}),0,\ldots,0) \\
& = & \max(\algintc{\bterm_i}(\algintc{c_1},\ldots,\algintc{c_k}),
  \algintc{c_1}(\vec{0}),\ldots,\algintc{c_k}(\vec{0})) \\
& = & \algintc{\app{\bterm_i}{\vec{c}}} \\
\end{array}
\]

To see that the given interpretation indeed orients all rules:
\[
\begin{array}{lcccccl}
\algintc{\minimum(\avar,\nul)} & = & 0 & \geq & 0 & = &
  \algintc{\nul} \\
\algintc{\minimum(\nul,\avar)} & = & 0 & \geq & 0 & = &
  \algintc{\nul} \\
\algintc{\minimum(\suc(\avar),\suc(\bvar))} & = & 0 & \geq &
  3 \cdot 0 & = & \algintc{\suc(\minimum(\avar,\bvar))} \\
\algintc{\diff(\avar,\nul)} & = & \avar & \geq & \avar & = &
  \algintc{\avar} \\
\algintc{\diff(\nul,\avar)} & = & \avar & \geq & \avar & = &
  \algintc{\avar} \\
\algintc{\diff(\suc(\avar),\suc(\bvar))} & = & 3 \cdot \avar +
  3 \cdot \bvar & \geq & \avar + \bvar & = &
  \algintc{\diff(\avar,\bvar)} \\
\algintc{\gcd(\suc(\avar),\nul)} & = & 3 \cdot \avar & \geq &
  3 \cdot 0 & = & \algintc{\suc(\nul)} \\
\algintc{\gcd(\nul,\suc(\avar))} & = & 3 \cdot \avar & \geq &
  3 \cdot 0 & = & \algintc{\suc(\nul)} \\
\algintc{\gcd(\suc(\avar),\suc(\bvar))} & = & 3 \cdot \avar +
  3 \cdot \bvar & \geq & \avar + \bvar + 3 \cdot 0 & = &
  \algintc{\gcd(\diff(\avar,\bvar),\suc(\minimum(\avar,\bvar)))} \\
\algintc{\build(\nul)} & = & 3 \cdot 0 & \geq & 0 & = &
  \algintc{\nil} \\
\algintc{\build(\suc(\avar))} & = & 9 \cdot \avar & \geq &
  4 \cdot \avar + 3 \cdot \avar & = & \algintc{
  \cons(\abs{\bvar}{\gcd(\bvar,\avar),\build(\avar)})} \\
\algintc{\collapse(\nil)} & = & 0 & \geq & 0 & = & \algintc{\nul} \\
\algintc{\collapse(\cons(F,t))} & = & F(t)+t & \geq &
  \max(F(t),t) & = & \algintc{\collapse(t)} \\
\end{array}
\]

\subsection{Proofs for Section~\ref{sec:polynomial}}

The proofs in Section~\ref{sec:polynomial} were somewhat minimal.
Here follow the complete proofs.

\begin{lemma}{Lemma~\ref{lem:polynomial:weak}}
If $p \in \Pol(\{ \avar_1 : \atype_1, \ldots, \avar_n : \atype_n
\})$, then $\fatlambda \avar_1 \ldots \avar_n.p$
$\in \WM_{\atype_1 \ftypepijl \ldots \ftypepijl
\atype_n \ftypepijl \abasetype}$.
\vspace{-2ex}
\end{lemma}

\begin{proof}
First note (**): \emph{$\fatlambda n m.n+m$ and $\fatlambda n m.n
\cdot m$ are weakly monotonic functionals in $\WM_{\abasetype_1
\ftypepijl \abasetype_2 \typepijl \abasetype_3}$} for any three
base types $\abasetype_1,\abasetype_2,\abasetype_3$.  This is easy
to see (take into account that $\WM_{\abasetype_1 \ftypepijl
\abasetype_2 \typepijl \abasetype_3}$ just consists of those
functions $f$ in the function space $\basealgebraset \functionpijl
\basealgebraset \functionpijl \basealgebraset$ such that $f(a,b)
\geq f(a',b')$ if $a \geq a'$ and $b \geq b'$).

The lemma holds by induction on the derivation of $p \in
\Pol(\{ \vec{\avar} \})$.

If $p \in \N$, then $\fatlambda \vec{\avar}.p$ is a
constant function; its weak monotonicity was demonstrated in
Example~\ref{exa:wmf}(\ref{ex:wmf:constants}).

If $p_1,p_2 \in \Pol(\{\vec{\avar}\})$, then by the induction
hypothesis $\fatlambda \vec{\avar}.p_1$ and $\fatlambda \vec{\avar}.
p_2$ are both weakly monotonic functionals.
Consider the $\lambda$-term $L := \abs{\bvar_1 \ldots \bvar_n}{
\app{\app{A}{(\app{F_1}{\vec{\bvar}})}}{(\app{F_2}{\vec{\bvar}})}}$,
and let $\varvaluation = \{ A \mapsto \fatlambda n m.n+m,\ F_1
\mapsto \fatlambda \vec{\avar}.p_1,\ F_2 \mapsto \fatlambda
\vec{\avar}.p_2\}$.
By Lemma~\ref{lem:algintfacts}(\ref{lem:algintfacts:interprete}),
$\lamalgint{L}_{\varvaluation} = \fatlambda \vec{\avar}.p_1 + p_2$ is
a weakly monotonic functional.
In the same way (using a valuation with $A \mapsto \fatlambda n m.
n \cdot m$), $\fatlambda \vec{\avar}.p_1 \cdot p_2$ is a weakly
monotonic functional.

Finally, suppose $p_1 \in \Pol^{\btype_1}(\vec{\avar}), \ldots, p_m
\in \Pol^{\btype_m}(\vec{\avar})$, and $\avar_i$ has type
$\btype_1 \typepijl \ldots \typepijl \btype_m \typepijl \abasetype$.
Since we can write $p_i = \fatlambda \vec{\bvar}.p_i'$ with
$p_i' \in \Pol^{\btype_i}(\vec{\avar},\vec{\bvar})$, the induction
hypothesis tells us that each $\fatlambda \vec{\avar}.p_i \in
\WM_{\vec{\atype} \ftypepijl \vec{\btype} \ftypepijl \abasetype}$.
Thus, $\fatlambda \vec{\avar}.\avar_i(\vec{p}) =
\lamalgint{\abs{\vec{\avar}}{\app{\avar_i}{(\app{\cvar_1}{\vec{\avar}})}
\cdots (\app{\cvar_m}{\vec{\avar}})}}_{\{\cvar_1 \mapsto p_1, \ldots,
\cvar_m \mapsto p_m \}}$ is a weakly monotonic functional by
Lemma~\ref{lem:algintfacts}(\ref{lem:algintfacts:interprete}).
\end{proof}

\begin{lemma}{Lemma~\ref{lem:strongpolynomial}}
Let $P(\avar_1,\ldots,\avar_n)$ be a higher-order polynomial of the
form $p_1(\vec{\avar}) + \ldots + p_m(\vec{\avar})$, where all $p_i(
\vec{\avar})$ are higher-order monomials.  Then $\fatlambda \vec{
\avar}.P(\vec{\avar})$ is strongly monotonic in argument $i$ if there
is some $p_j$ of the form $a \cdot
\avar_i(\vec{b}(\vec{\avar}))$, where $a \in \N^+$.
\end{lemma}

\begin{proof}
Let weakly monotonic functionals $N_1,\ldots,N_n,M_i$ be given, and
some $i,j$ such that $p_j = a \cdot \avar_i(\vec{b}(\vec{\avar}))$
with $a \in \N^+$.  It suffices to see that if $N_i \gwm M_i$, then
also $P(N_1,\ldots,N_i,\ldots,N_n) > P(N_1,\ldots,M_i,\ldots,N_n)$.
In the following, $\vec{N}$ is short notation for $N_1,\ldots,N_i,
\ldots,N_n$ and $\vec{N'}$ is short notation for $N_1,\ldots,M_i,
\ldots,N_n$.

By Lemma~\ref{lem:polynomial:weak}, $\fatlambda \vec{\avar}.p_k(
\vec{\avar})$ is a weakly monotonic functional for all $k$, and this
implies that $p_k(\vec{N}) \geq p_k(\vec{N'})$.  If, moreover,
$p_j(\vec{N}) > p_j(\vec{N'})$, then we obtain $P(\vec{N}) >
P(\vec{N'})$, as required, by the nature of the addition operator.

Write $p_j(\vec{\avar}) = a \cdot d(\avar_i,\vec{\avar})$, where
$d(\bvar,\vec{\avar}) = \bvar(\fatlambda \vec{\bvar_1}.b_1(
\vec{\avar},\vec{\bvar}), \ldots, \fatlambda \vec{\bvar_m}.
b_m(\vec{\avar},\vec{\bvar}))$.
Since all $b_j$ are polynomials, Lemma~\ref{lem:polynomial:weak}
provides that $\fatlambda \vec{\bvar_j}.b_j(\vec{N},\vec{\bvar}) \geq
\fatlambda \vec{\bvar_j}.b_j(\vec{N'},\vec{\bvar})$, so by weak
monotonicity of $N_i$ we know $d(N_i,\vec{N}) \geq d(N_i,\vec{N'})$.
By the definition of $N_i \gwm M_i$, we also see that
$d(N_i,\vec{N'}) > d(M_i,\vec{N'})$.
Since for $a > 0$ we have $a \cdot k > a \cdot j$ if $k >
j$, it follows that $p_j(\vec{N}) = a \cdot d(N_i,\vec{N}) > a \cdot
d(M_i,\vec{N'}) = p_j(\vec{N'})$ as required.
\end{proof}

} % \confreport
\end{document}